\documentclass[12pt]{article}
\pdfoutput=1
\usepackage{amsthm,amsopn,algorithm, algorithmic, amscd, amsfonts, amsmath, amssymb}

\usepackage[affil-it]{authblk}
\usepackage{titling}
\usepackage{caption}
\usepackage{color}
\usepackage[pdftex]{graphicx}
\usepackage{bm}
\usepackage{float}
\usepackage{natbib}
\usepackage{enumitem}
\usepackage{rotating}
\usepackage[nokeyprefix]{refstyle}
\usepackage{varioref}
\usepackage{xr-hyper}
\usepackage{hyperref}
\usepackage{bbm}

\numberwithin{equation}{section}
\newlist{longenum}{enumerate}{1}
\setlist[longenum,1]{label=\alph*)}
\DeclareMathOperator*{\argmax}{argmax}

\usepackage{setspace}
\usepackage{titling}
\usepackage{caption}
\usepackage{color}
\usepackage[pdftex]{graphicx}
\usepackage{bm}
\usepackage{float}
\usepackage{natbib}

\usepackage{commonstuff}

\title{Extracting Common Time Trends from Concurrent Time Series: Maximum Autocorrelation Factors with Applications}
\author{Matz A. Haugen%
  \thanks{Electronic address: \texttt{mahaugen@stanford.edu}}}
\author{Bala Rajaratnam}
\author{Paul Switzer}
\affil{Stanford University}

\begin{document}

\maketitle
\abstract{

Concurrent time series commonly arise in various applications, including when monitoring the
environment such as in air quality measurement networks, weather
stations, oceanographic buoys, or in paleo form such as lake
sediments, tree rings, ice cores, or coral isotopes, with each
monitoring or sampling site providing one of the time series. The goal
in such applications is to extract a common time trend or signal in the
observed data. Other examples where the goal is to extract a common
time trend for multiple time series are in stock price time series,
neurological time series, and
quality control time series. For this purpose we develop properties of MAF
[Maximum Autocorrelation Factors] that linearly
combines time series in order to maximize the resulting
SNR [signal-to-noise-ratio] where there are multiple smooth signals present in the data. Equivalence is established in a regression setting between MAF and CCA [Canonical Correlation Analysis] even though MAF does not require specific signal knowledge as opposed to CCA. We proceed to derive the theoretical properties of MAF and quantify the SNR advantages of MAF in
comparison with PCA [Principal Components Analysis], a commonly used
method for linearly combining time series, and compare their
statistical sample properties. MAF and PCA are then applied to real and simulated data sets to illustrate MAFs efficacy.
}
\section{Introduction and Preliminaries}
\label{sec:intro}
A common goal in the analysis of a collection of $p$ concurrent time
series $Z_j(t)$, $j=1, \dots, p$, observed at times $t=1,\dots, n$, is
to extract a common time trend which we refer to as the signal. Specifically,
we look at optimizing a linear combination $\bm Y(t) = \bm w' \bm
Z(t)$, $t=1, \dots, n$, where $\bm w$ is an optimized coefficient $p$-vector. For example, if the goal is to maximize variance over time of the
combined series, $Y(t)$, then this is equivalent to finding the
first principal component in a PCA (Principal Component
Analysis). Then the coefficient vector $\bm w_{PCA}$ is the principal eigenvector
of the cross-covariance matrix, $\bm S$, where $\bm S_{ij}$  is the covariance over time between the
pair of time series $Z_i(t)$ and $Z_j(t)$. The idea of PCA is to reduce dimensionality through retaining linear combinations of the data which have the highest variability. Some applications of PCA to
multiple time series analysis are given in \citet{li07,
  Briffa2008,Mcshane,JB2012} Find references outside earth sciences. However, maximizing variance across time, as PCA seeks to do, will not
  necessarily be well suited to revealing coherent underlying latent
  time trends because PCA does not make use of the specific time order
  of the data or optimize any property dependent on temporal coherence. If the time order of the time series were permuted,
say, then the covariance matrix $\bm S$ and the coefficient vector
$\bm w_{PCA}$ are unchanged.

Arguably, an optimization criterion for the coefficient vector $\bm w$
for combining the $p$ concurrent time series should
specifically maximize a measure of temporal coherence of the
transformed time series, rather than the time variance used in PCA.

\subsection{MAF - Maximum Autocorrelation Factors}

 An alternative to PCA is Maximum Autocorrelation Factors (MAF) \citep{switzer84, switzer89} where variance
maximization is replaced by autocorrelation maximization, which
explicitly does depend on the time ordering of the $p$-variate
observations. The motivation for MAF is that smoothly evolving time
trends contained in time series data will enhance autocorrelation. We show in Appendix
\ref{app:proofs} that the
MAF-optimized coefficient vector $\bm w_{MAF}$ is obtained as the
leading eigenvector of the matrix

\begin{equation}
  \label{eq:15}
 \bm S^{-1/2} \bm S_{\Delta}
\bm S^{-1/2},
\end{equation}
where $\bm S_\Delta$ is the $p\times p$ covariance matrix of
the time-differenced time series. Any rescaling of the original time
series, $\bm Z(t)$, will preserve the MAF time series. This invariance property for MAF is also derived in Appendix
\ref{app:proofs}. On the other hand, PCA component time series are not invariant to rescaling or recombining of the original data.

Some applications of MAF to
multiple time series analysis are given in \citet{switzer84,
  switzer89, Gallagher2014}. Our interest in MAF derives from
applications to the analysis of multiple time series of climate proxy
data from tree ring measurements, described in
\secRef{sec:applications}. A fuller discussion of the analysis of tree
ring data will be presented in a separate paper. {In this paper we shall focus on the methodological development of the MAF framework.}

To intuitively appreciate the difference between MAF and PCA, suppose
we have $p=2$ time series, one that is pure white noise and the other
that is a linear time trend without noise, with both series having unit
variance over time. {Since PCA looks for a combined time series with maximal variance, it} is indifferent between the noisy time
series with zero autocorrelation and the clean time series with unit
autocorrelation. {On the other hand}, MAF will put all its weight on the noiseless
linear time trend. If the two original time series contained each a
mixture of time trend and noise, then the MAF time series will amplify
the time trend relative to the noise.

\subsection{An Illustration}

Figure \ref{simpleExample} shows an example with four parallel time series, rescaled to
have zero mean and unit variance. These 150-year time series are
extracted from the database used in \cite{Mann2008} and represent tree-ring time series. {To measure temporal coherence we introduce an empirical signal-to-noise ratio (SNR), which is obtained by taking the ratio of two standard deviations; that of a smoothed version of the time series and that of the associated residuals after the smooth has been subtracted from the original. Standard deviations are calculated by summing over the time steps.  The annotated empirical SNR} suggest
that the first two time series exhibit more evident temporal structure
than the last two time series. The corresponding PCA and MAF time
series are shown in \fig{simpleExampleMAFs}, and these are also rescaled to have
zero mean and unit variance. The MAF time series appears to
concentrate the temporal structure whereas PCA seems to exhibit more
temporal noise. The
empirical SNR of the MAF time series is $1.46$ while that of the
PCA time series is $0.92$. PCA and MAF coefficient matrices are shown
in Table 1 and we see that the MAF time series up-weights the first two data time series and
down-weights the last two data time series.

\EPSFIG[scale=0.60]{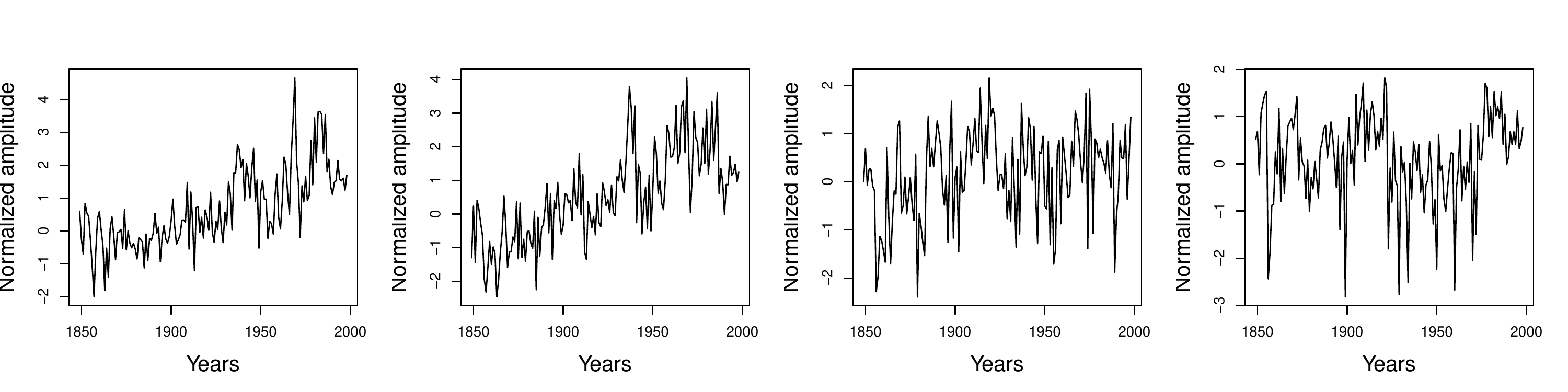}{Four tree ring time series,
  each one scaled to have unit
  variance and zero mean. Autocorrelation is annotated above each figure.}{simpleExample}

\EPSFIG[scale=0.65]{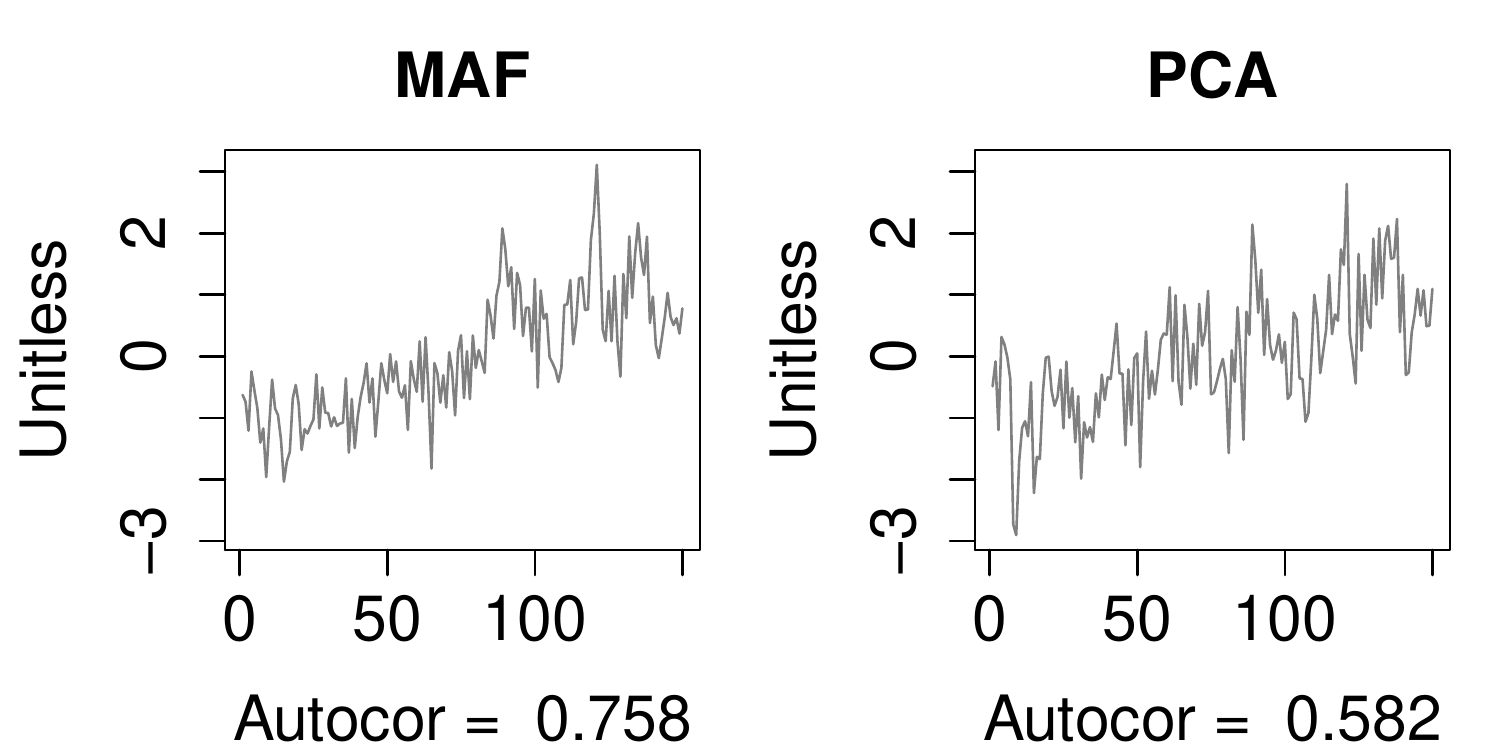}{MAFs and PCs of the time
  series shown in \fig{simpleExample}.}{simpleExampleMAFs}

\begin{table}[ht]
\centering
\begin{tabular}{rrr}
  \hline
 & MAF & PCA \\
  \hline
1 & 0.80 & 0.59 \\
  2 & 0.30 & 0.58 \\
  3 & 0.24 & 0.42 \\
  4 & -0.47 & 0.37 \\
   \hline
\end{tabular}
\caption{MAF and PCA coefficients of 4 time series.}
\label{MAFPCACoefs}
\end{table}

\subsection{Summary of results}
In \secRef{sec:snModel}, we introduce the signal-plus-noise model and
show that under general conditions, {the MAF time series
yields the highest signal-to-noise ratio among all possible combined time series}. Equivalently, MAF also maximizes the correlation
between the combined time series and the underlying signal time
series. The PCA time series, on the other hand, maximizes
signal \emph{plus} noise variance rather than the ratio. PCA does not {generally} share the MAF ``oracle property'', {i.e. finding the linear combination of time series which is maximally correlated with the underlying signal}. We show that the SNR of the MAF time series is equal or greater than that of the PCA time series in all situations involving one {or more} signals. Only in the trivial setting when the noise is \emph{iid}, i.e. with zero cross-correlation and equal variance are MAF and PCA equivalent. Otherwise, MAF
{increases} the SNR compared to PCA.

{We then extend the model to having $q\leq p$ multiple signals, where we establish that first $q$ MAFs and Canonical Correlation Factors (CCFs) span the same subspace that contains any linear combination of the underying signal time series, thus extending the ``oracle property'' to the case of multiple signals. Consequently, in a regression setting with one response time series and a set of predictor time series, where the latter contains multiple signals, MAF regression with $q$ factors will be optimal in a `least squares' sense. It is assumed that the response signal is a particular linear combination of the underlying set of signals present in the predictors. On the other hand, since the first $q$ Principal Components (PCs) do not span the subspace of signals, their regression on the response will be suboptimal in the least squares sense.}

In \secRef{sec:illustrations}, a specific
illustration is given where two groups of time series are considered,
each with different signal strengths present in combination with noise. Explicit expressions are given for both
MAF and PCA where we replace the sample covariance matrices by their expected values under ther model.  We then derive the explicit form of the MAF and PCA
coefficients vectors in {other models. Doing so allows us to investigate how the coefficients change as functions of the noise cross-correlation, relative signals strengths contained in each time series, and total number of time series. We find that the leading MAF SNR improves compared to PCA as noise cross-correlation, number of time series and/or signal strength differences increase(s).}

\secRef{sec:samplingProperties} explores the
statistical properties of MAF and shows that MAF coefficient
estimates are consistent as the number of time steps are increased
while keeping the number of time series constant. Illustrations are
also given to quantify the difference between MAF and PCA regarding
their correlations with the underlying time trend. To determine the
presence of a signal in the data, a hypothesis testing procedure is
presented where the null hypothesis is a pure noise time series. Using resampling, we illustrate the power of the
test at different sample sizes and significance levels.

Application to tree ring time series in western North
America is shown in \secRef{sec:applications}. We illustrate MAF and
PCA for these time series. Both MAF and PCA suggest underlying
common time trends, but MAF appears to show these trends more
clearly. A null hypothesis
test is highly significant and suggests the presence of time trends in
the data. Concluding remarks are presented in
\secRef{sec:conclusion}.


\section{The signal-plus-noise model}
\label{sec:snModel}
\subsection{Preliminaries}
We now formally define the Maximum Autocorrelation Factor (MAF). For a given set of $p$ observed concurrent time series, $\bm Z(t): \mathbb{N} \rightarrow \R^p$, the leading MAF coefficient vector, is defined as the linear combination of the time series in $\bm Z(t)$ such that
\begin{equation}
\label{eq:prelim1}
  \bm w_{MAF}(\bm Z) \equiv \bm w_{MAF} = \argmax_{\bm w\in \R^p} \{\text{Cor}(\bm w' \bm Z(t), \bm w' \bm Z(t+1))\}.
\end{equation}
Similary, the leading Principal Component coefficient vector are defined as
\begin{equation}
 \bm w_{PCA}(\bm Z) \equiv \bm w_{PCA} =  \argmax_{\bm w\in \R^p} \{\text{Var}(\bm w' \bm Z(t))\}.
\end{equation}

Note that the MAF yields the optimal linear combination such that the autocorrelation is maximized while the PC yields the linear combination that maximizes variance. Furthermore, the leading MAF factor is defined as follows,
\begin{equation}
   \bm Y_{MAF}(\bm Z (t)) = \bm w'_{MAF}\bm Z(t) \text{ for } t = 1,..., n.
 \end{equation}
This single time series is the linear combination of the original $p$ time series with maximal autocorrelation. With these definitions, we now proceed to derive various properties related to these two techniques.
\subsection{The model}
Suppose $f(t): \mathbb{N} \rightarrow \R$ is a fixed but unknown normalized underlying signal time series
with zero mean and {Euclidian norm} equal to 1 over the observation period $t=1,
..., n$. We have $p$ observed concurrent time series, $\bm Z(t): \mathbb{N} \rightarrow \R^p$, that are represented as

\begin{equation*}
\bm Z(t) =  \bm s(t) + \bm\varepsilon(t) = f(t) \cdot \bm b  + \bm\varepsilon(t),
\end{equation*}
\begin{equation}
\text{i.e. }\quad \sum_tf(t) =  0 \text{  and  } \sum_tf^2(t) = 1,
\label{eq:model}
\end{equation}
\begin{equation*}
\text{with} \: E[\bm\varepsilon(t)] = 0 \:\text{and
  } Var[\bm\varepsilon(t)] = \bm \Sigma_\varepsilon, \quad \forall t
\end{equation*}
 where $\bm \varepsilon(t)$ is the random $p$-variate covariance-stationary residual noise time series and $\bm b = (b_1,b_2, \dots, b_p)$ is a
 coefficient vector, {fixed and unknown}. The quantities
 $f(t)$, $\bm \varepsilon(t)$ and $\bm b$ are all unobserved and
 unknown. We call this the `S+N model'. A linear combination of the $p$
observed times series $\bm Z(t)$ is another time series $\bm w'\bm
Z(t)$, with $\bm w' \in \R^p$. The signal-to-noise ratio for the combined time series is
denoted $SNR(\bm w)$ and defined as
\begin{align}
\label{eq:R}
\frac{\text{Average signal mean
    square}}{\text{Average noise variance}} =& \text{SNR}(\bm w)
      := \frac{\frac{1}{n} \sum_{t=1}^n Var[\bm w' \bm s(t)]} {\frac{1}{n} \sum_{t=1}^n Var[\bm w' \bm \varepsilon(t)]} \nonumber\\
      =& \frac{\frac{1}{n}\sum_{t=1}^n f^2(t)[(\bm w' \bm b)^2]}{\frac{1}{n}\sum_{t=1}^n\bm w'\Sigma_\varepsilon \bm w} \nonumber\\
      =& \frac{(\bm w'\bm b)^2}{\bm w'\Sigma_\varepsilon \bm w}.
\end{align}

The MAF and PCA time series are examples of such linearly combined
time series with particular choices for $\bm w$.

Now define
\begin{equation}
\label{eq:kfke}
 \frac{1}{n-1}\sum_{t=1}^{n-1}\varepsilon(t)\varepsilon(t+1) = k_\varepsilon, \;\;\text{and  } \frac{1}{n-1}\sum_{t=1}^{n-1}f(t)f(t+1) = k_f,
\end{equation}
where $k_f$ may be regarded as a
measure of signal coherence. We now show
conditions under which the MAF time series maximizes SNR$(\bm w)$ over
$\bm w$.

\begin{prop}
\label{propMonotone}
Suppose that the stationary time series model for
the residual noise is such that the $p\times p$ residual
autocovariance matrix has the proportional form,
\begin{align}
  \label{eq:11}
 \text{Cov}(\bm\varepsilon(t), \bm\varepsilon(t+1))  =
 k_\varepsilon \text{Cov}(\bm\varepsilon(t)) =
k_\varepsilon \bm \Sigma_\varepsilon, \quad \text{for some }k_\varepsilon \in \R,
\end{align}
such that
\begin{equation}
  \label{eq:17}
  \text{Var}(\bm w' \bm Z(t), \bm w' \bm Z(t+1))  =  (\bm w'\bm b)^2 k_f  +
k_\varepsilon \bm w'\Sigma_{\bm \varepsilon} \bm w , \quad\text{with } k_f > k_\varepsilon,
\end{equation}
where  $k_f$ is the lag-1 autocorrelation of a normalized signal
 $f(t)$, as given in Equation \ref{eq:kfke}. Then MAF maximizes S/N and PCA maximizes S+N, i.e.,
\begin{align}
   \bm w_{MAF} =& \argmax_{\bm w \in \R^p} \frac{(\bm w' \bm b)^2}{\bm w' \bm \Sigma_\varepsilon \bm w}, \;\; \text{and} \nonumber\\
   \bm w_{PCA} =& \argmax_{\bm w \in \R^p} \{(\bm w' \bm b)^2 + \bm w' \bm \Sigma_\varepsilon \bm w\}.
\end{align}
\end{prop}

\begin{proof}
 We show that maximizing SNR$(\bm w)$ over linear
combinations if $\bm w$ is equivalent to maximizing the lagged autocorrelation, denoted $r(\bm w)$, of the combined time series $\bm w'\bm Z(t)$. Now define the following,
\begin{align}
\label{eq2}
r(\bm w) =& \frac{Cov(\bm w' \bm Z(t), \bm w' \bm Z(t+1)}{Cov(\bm w' \bm Z(t))} =
1 -  \frac{\text{Var}(\bm w'\Delta \bm Z(t))}{2 \text{Var}(\bm w'\bm
  Z(t))}, \nonumber \\
\Delta \bm Z(t) =&  \bm Z(t)-\bm Z(t-1) = \bm b' \Delta f(t) + \Delta
\bm\varepsilon(t) \quad\text{is the time differenced data vector,}\nonumber \\
\Delta f(t) =& f(t)-f(t-1) \quad\text{is the time-differenced signal,}
\nonumber \\
\Delta
\bm\varepsilon (t) =& \bm\varepsilon(t)-\bm\varepsilon(t-1)
\quad\text{is the time differenced noise vector.}
\end{align}

We can write
\begin{align}
\label{eq3}
\text{Var}(\bm w’\bm Z(t))   =& (\bm w'\bm b)^2  +  \bm w' \Sigma_{\bm
  \varepsilon} \bm w.
\end{align}

 Using \eqref*{eq:R}, \eqref*{eq2}, \eqref*{eq3} we can express the
 model autocorrelation, $r(\bm w)$, of the combined time series $\bm w'\bm Z(t)$ as
\begin{equation}
\label{eq:monotoneEq}
r(\bm w) =  \frac{\text{SNR}(\bm w) \cdot k_f + k_\varepsilon}{\text{SNR}(\bm
  w) + 1}
\end{equation}
which is a monotone
function of SNR$(\bm w)$, if $k_f >
k_\varepsilon$. Hence, maximizing $r(\bm w)$ is
equivalent to maximizing SNR$(\bm w)$. Since MAF maximizes
autocorrelation, MAF will also maximize the signal-to-noise variance ratio over combinations of $p$ observable cross-correlated time series, where each observable time series is a sum of a signal contribution and a random noise contribution.

PCA, one the other hand, is defined as
\begin{equation}
  \label{eq:43}
  \bm w_{PCA} := \argmax_{\bm w\in \R^p} \{\text{Var}(\bm w' \bm Z(t))\} = \argmax_{\bm w\in \R^p} \{Var[\bm w'(\bm b'f(t) + \bm \varepsilon(t))]\} =   \argmax_{\bm w\in \R^p} \{(\bm w'\bm b)^2 + \bm w' \Sigma_\varepsilon \bm w\}.
\end{equation}
\end{proof}

The above theorem has important consequences. In signal extraction, maximizing SNR$(\bm w)$ is arguably more
desirable than maximizing overall variance of a linear combination of the
input time series as in PCA. The MAF optimization criterion is clearly
more suited to the goal of extracting a common signal component from
multiple time series. It is also important to note that the
MAF time series is invariant to any rescaling of the input time
series, shown in Appendix \ref{app:proofs}, whereas the PCA time series is scale dependent.

We now proceed to state the theoretical properties of MAF time series in terms of four lemmas. First, we show that the MAF time series is maximally correlated with the underlying
  signal time series under the S+N model. This property is fundamentally important and is henceforth referred to as the ``oracle property'' of MAF.

\begin{lem}
\label{lem:monotone}
Consider the model given in Proposition \ref{propMonotone}, then
  \begin{equation}
    \bm w_{MAF} = \argmax_{\bm w \in \R^p} \;\;Cor[f(t), \bm w' \bm Z(t)]
  \end{equation}
\end{lem}
\begin{proof}
  First note that the squared cross-correlation is given by
  \begin{equation}
    \label{eq:8}
Cor[f(t), \bm w' \bm Z(t)]^2 = \frac{\text{SNR}(\bm w)
}{\text{SNR}(\bm w) + 1}.
  \end{equation}
The proof now follows immediately by letting $k_f = 1$ and $k_\varepsilon = 0$ in \eq{eq:monotoneEq} in the proof of Proposition \ref{propMonotone}.
\end{proof}
{\it Remark:} Note that Lemma \ref{lem:monotone} above is incidentally the defining property of Canonical Correlation Analysis (CCA) with one underlying signal. However, there is a fundamental difference: CCA requires the knowledge of the signal, $f(t)$, while MAF does not, hence the above lemma being called the ``oracle property'' of MAF.\\

We now proceed to show that the leading MAF time series is invariant to any rescaling of the input time series.
\begin{lem}
\label{lem:transform}
Consider a data matrix $\bm Z \in R^{n\times p}$ where each column of $\bm Z$ represents a single time series of length $n$ and $Y_{MAF} (\bm Z)$ represents the leading MAF factor of $\bm Z$.
Now let $\bm A$ be an invertible matrix such that $\tilde{\bm Z} = \bm Z \bm
A$. Then,
\begin{equation}
 Y_{MAF} (\bm Z) = Y_{MAF} (\tilde{\bm Z}).
\end{equation}
\end{lem}
\begin{proof}
We shall show,
\begin{equation}
 Y_{MAF} (\bm Z) := \bm Z \bm w_{MAF}(\bm Z) = \tilde{\bm Z}  \bm w_{MAF}(\tilde{\bm Z}) =: Y_{MAF} (\tilde{\bm Z}).
\end{equation}
First from \eq{eq:prelim1},
\begin{align}
  \bm w_{MAF}(\bm Z) =& \argmax_{\bm w\in \R^p} \{\text{Cor}(\bm w' \bm Z(t), \bm w' \bm Z(t+1))\} \nonumber \\
 =& \argmax_{\bm w \in \R^p} \frac{\bm w' \bm S_\delta \bm w}{\bm w' \bm S \bm w}
\end{align}
where $\bm S_\delta = Cov(\bm Z(t), \bm Z(t+1))$ and $\bm S = Cov(\bm Z(t))$. Then note,
\begin{align}
  \bm w_{MAF}(\bm Z \bm A) =& \argmax_{\bm w \in \R^p} \frac{\bm w' \bm A' \bm S_\delta \bm A \bm w}{\bm w' \bm A' \bm S \bm A \bm w} \nonumber\\
   =& \argmax_{\bm u \in \R^p} \frac{\bm u' \bm S_\delta \bm u}{\bm u' \bm S \
   \bm u}, \text{with } \bm u = \bm A \bm w,
\end{align}
as $\bm A$ is invertible. The above then gives
\begin{equation}
  \bm w_{MAF}(\bm Z \bm A) = \bm A^{-1} \bm w_{MAF}(\bm Z).
\end{equation}
Thus,
\begin{equation}
 Y_{MAF} (\bm Z) =: \tilde{\bm Z}\bm w_{MAF}(\tilde{\bm Z}) = \bm Z \bm A \bm A^{-1} \bm w_{MAF}(\bm Z) = \bm Z \bm w_{MAF}(\bm Z) =: Y_{MAF} (\tilde{\bm Z}).
\end{equation}
\end{proof}

We now proceed to give an analytic representation of the MAF coefficient vector.

\begin{lem}
\label{lem:expectedLikelihood}
Consider the `S+N model' in \eq{eq:model}. It follows that the MAF coefficient vector can be expressed as
  \begin{equation}
    \label{eq:10}
    \bm w_{MAF} = \bm \Sigma_\varepsilon^{-1} \bm b,
  \end{equation}
\end{lem}
\begin{proof}
  See proof of Lemma \ref{lem:expLikeli}.
\end{proof}

Lastly, we show that the SNR of the MAF time series under the S+N model is proportional to the expected value of a likelihood ratio statistic for a Gaussian noise specification.

\begin{lem}
\label{lem:expLikeli}
Consider the `S+N model' in \eq{eq:model} and the following set of hypotheses,
\begin{align}
\label{eq:hypTest2}
  H_0: &   \bm Z_n(t) = \bm b + \bm \varepsilon(t) \nonumber\\
H_A: & \bm Z_n(t) = f(t) \bm b + \bm \varepsilon(t), \quad f(t) \neq
 \text{constant},
\end{align}
such that, $\sum_tf(t) = 0$ and $\sum_tf^2(t) = 1$ as defined in
\eqref*{eq:model}, and $\bm \varepsilon(t) \sim N(0, \bm\Sigma)$. Then,
\begin{equation}
  SNR(\bm w_{MAF}) = E[\ln(L_A) - \ln(L_0)],
\end{equation}
where $L_A$ and $L_0$ are the likelihoods of the two hypotheses given the data matrix $\bm Z \in \R^{n \times p}$.
\end{lem}
\begin{proof}
  See Appendix \ref{app:proofs}.
\end{proof}

\citet{switzer84} show that MAF and PCA are equivalent in the special and restrictive case where the noise covariance matrix is given by
\begin{equation}
  \label{eq:44}
  Cov(\bm \varepsilon(t)) = \Sigma_\varepsilon = \sigma^2 \bm I,
\end{equation}
i.e., the noise component of each input has the same variance and these
$p\times p$ noise components have no cross-correlation. However, this
equivalence between MAF and PCA
does not hold when there is noise cross-correlation or heterogeneous noise
variance.

\subsection{Multiple signals model}
We can generalize the S+N model to allow for multiple underlying signal
time series. Each of the $p$ observed concurrent time series is made up of its own unknown smooth signal time series and its own superposed noise time series representing short term fluctuations. The specific structure of the problem represents each of these $p$ signal time series in terms of $q<p$ underlying orthogonal factor time series, representing the reduced dimensionality of the signal structure.  The goal is to find $q$ new time series which are linear combinations of the observed time series. These $q$ new time series aim to recover the underlying orthogonal factor time series, i.e. the signals.  We show conditions under which the MAF linear combinations of the observed time series achieve this concentration of the underlying signal information.

Our strategy for showing that the $q$ MAF time series jointly capture the available signal information contained in the observed $p$-variate time series is to demonstrate that the $q$-space spanned by MAF is the same as the $q$-space obtained from a canonical correlation analysis (CCA) of the $p$ observed times series one the one hand and the $q$ unobserved signal time series on the other hand. Theorem \ref{multipleSignalsTheorem} below shows this equivalence, under specific conditions for the additive noise component of the observed time series. The equivalence, using the modeled noise covariance structure, implies that MAF, which is computed \emph{without} specifying the underlying signal, can capture the same signal information as a canonical correlation analysis which \emph{requires} the signal specification. In this sense MAF can be said to have an oracle property under the specified conditions insofar as covariances and lagged covariances computed from the observed data approximate their modeled structure. Thus, MAF is able to achieve the same result as CCA by taking advantage of time order.

Consider a $p$-variate set of time series, $\bm Z(t), \: t = 1,\dots n > p$, comprised of $q \leq p$ normalized underlying smooth orthogonal signals, $\bm F(t) = (f_1(t), f_2(t), ..., f_q(t))$, and zero-mean $p$-variate noise, $\bm\varepsilon(t) \sim \bm \Sigma_\varepsilon$. For the signal, we assume
\begin{equation}
 \frac{1}{n-1}\sum_{t=1}^{n-1}f_i(t)f_j(t+1) = k_i \delta_{ij}
\end{equation}
where $1 > k_1 \geq k_2 \geq \dots \geq k_q > k_\varepsilon$ {and $\delta_{ij}$ us the familiar Kronecker delta function}\footnote{We neglect any non-orthogonality that might arise between lagged versions of the {signal} time series.}. For the noise, assume a proportional covariance model $Cov[\bm \varepsilon(t), \bm \varepsilon(t+1)] = k_\varepsilon \bm \Sigma_\varepsilon$.  A $p$-length signal strength vector, $\bm b_i$, describes the amount of signal $f_i(t)$ present in each of the $p$ original time series $\bm Z_j(t)$. With $\bm B \in \R^{(p\times q)}$ with columns $(\bm b_1, \bm b_2, \dots, \bm b_q) $, the full model is
\begin{align}
  \label{eq:93}
 \bm Z(t) =   \bm B \bm F(t) + \bm \varepsilon(t).
  \end{align}
  Letting $diag(\bm k)$ be the matrix formed by the $q$-vector $\bm k = (k_1, k_2, \dots, k_q)$ in the diagonal and zeros in the off diagonal, we can write
\begin{align}
  \bm \Sigma_Z =& Cov[\bm Z(t), \bm Z(t)] = \bm B \bm B' + \bm \Sigma_\varepsilon  \nonumber\\
  \bm \Sigma_{\delta \bm Z} =& Cov[\bm Z(t), \bm Z(t+1)] \nonumber\\
= & \bm B Cov[\bm F(t), \bm F(t+1)] \bm B' + Cov[\bm \varepsilon(t), \bm \varepsilon(t+1)]\nonumber\\
  = & \bm B diag(\bm k)\bm B'
    + k_\varepsilon \bm\Sigma_\varepsilon,  \nonumber\\
 \end{align}
both assumed to be positive definite.

{Canonical Correlation Analysis (CCA) looks for} linear combinations of the columns of $\bm Z$ which maximize correlation between linear
combinations of the signals contained in $F(t)$, while being orthogonal to each other. {We shall refer to these combinations as Canonical Correlation Factors (CCFs).}

In Appendix \ref{app:proofs}, we show that the first $q$ MAF and CCA factor coefficients for $\bm Z$ are both contained in the range of $\bm \Sigma_\varepsilon^{-1} \bm B$, formalized in the following theorem.

\begin{thm}
\label{multipleSignalsTheorem}
If $\min_i(k_i) > k_\varepsilon$, the first $q$ CCA and MAF coefficient vectors span the same hyperplane of dimension $q$ in $\R^p$.
\end{thm}

Consequently, the first $q$ MAFs are optimal as regressors in the following sense. Since the first $q$ CCFs maximize correlations of $q$ different linear combinations of $f_1(t), \dots, f_q(t)$, we can construct a maximizer of any linear combination of $f_1(t), \dots, f_q(t)$. Moreover, for a response variable $y(t) = \sum_q \alpha_i f_i(t)$, there exists a linear combination of the CCFs, $\hat{y}(t)$, for which the corresponding correlation, Cor$(y(t), \hat{y}(t))$, is maximized. Thus, $\hat{y}(t)$ is an optimal predictor of $y(t)$ using $q$ time series in a least squares sense. And because MAF spans the same $q$-subspace as the first $q$ CCFs, by trasitivity, MAF is also optimal in this sense. The benefit of MAF is that no knowledge of the underlying signal is needed for its computation as opposed to CCA.

For PCA in the multiple signal case we find the eigenvectors of
\begin{equation}
\label{PCAsimple}
\bm \Sigma_Z = \bm B \bm B' + \Sigma_\varepsilon.
\end{equation}

If the data time series has been normalized by their respective variance, $\bm \sigma_Z^2 = \text{diag}(\bm \Sigma_Z)$ the diagonal of the covariance matrix, the corresponding normalized PCA would we the eigenvectors of
\begin{equation}
\text{diag}(\bm\sigma_Z)^{-0.5}\left[\sum_i \bm b_i \bm b_i' + \Sigma_\varepsilon\right]\text{diag}(\bm\sigma_Z)^{-0.5},
\end{equation}
where $\text{diag}(\bm\sigma_Z)^{-0.5}$ has the variance in the diagonal and zeros in the off-diagonal. In both cases there is no closed form for the PCA eigenvectors. Moreover, the space spanned by the first $q$ PC coefficient vectors are not the same as the space spanned by CCA, and thus PCA is sub-optimal in this setting. This can easily be seen by noting that $\bm \Sigma_\varepsilon^{-1} \bm b_i$ is not an eigenvector of either matrix. An important aspect of this sub-optimality comes from PCAs lack of invariance under linear transformations. Looking at \eq{PCAsimple}, the only situation in which MAF and PCA are equivalent is if $\Sigma_\varepsilon = \bm I$.

In a situation with no noise, MAF
will recover a multivariate mixture of orthogonal signals into their
separate components without loss of information. For example, if two time series are supplied,
both with a combination of a linear and a quadratic signal and both mutually orthogonal, then the MAFs will decompose these two into their separate forms.

\begin{property}
\label{pr:nonoise}
Let $\bm z(t) =  \bm B \bm F(t)$ and $\bm F(t)$ represents q
concurrent unknown and uncorrelated time series at time $t$ such that
$q \leq p$ sorted in decreasing autocorrelations, $k_1 \geq ... \geq
k_q$, $\bm k$ in vector form. And let $\bm B$ be an unknown $p \times q$
matrix. Then the MAF will recover $\bm F(t)$. If $q=p$, $\bm B$ will also be recovered.
\end{property}

\begin{proof}
See Appendix \ref{app:proofs}.
\end{proof}

\section{Illustrations of MAF/PCA comparisons  in S+N model}
\label{sec:illustrations}

We now consider two models where it is possible to derive closed form expressions for the SNRs of the leading MAF and PC. These expressions allow us to quantify the improvement that MAF yields over PCA, and get a firm understanding of how each model parameter affects the different SNRs. A closed form expression is also derived for the leading MAF coefficient, $\bm w_{MAF}$.

\subsection{Model I: Two groups of time series}
Consider a scenario with two groups of concurrent time series following the signal-plus-noise model, with time-independent noise. Both groups contain $q$ time series each. The following lemma gives the relationship between the coefficients of each group of time series both for the MAF and the PCA case.

\begin{lem} Consider two groups of $q$ time series, with SNR equal to $b_1$ and $b_2$, with $b_2<b_1$ respectively, and where noise has equal variance and a common
 cross-correlation of $\rho > \frac{-1}{p-1}$ between each $2q$ time series. Let the total number of time series be represented by the $2q$-vectors $\bm Z(t)$, for $t = 1, ..., n$. Then consider a linear combination of these $2q$ time series $\bm w' \bm Z(t)$. The associated SNR of this linear combination is given by
\begin{align}
  \label{eq:5}
  \text{SNR}(\bm w) = \text{SNR}(w_1, w_2) =&
  \frac{b_1^2q[1+\nu\gamma]^2}{(1-\rho)(1+\nu^2)+\rho q(1+\nu)^2}
  \nonumber\\
  \nu =& w_2/w_1, \quad \gamma = b_2/b_1,
\end{align}
where $w_1$ and $w_2$ represent the coefficient for each group of time series. The maximum SNR, and also the MAF SNR, occurs when
\begin{equation}
  \label{eq:33}
  \nu_{MAF} := \frac{w_2}{w_1}= \frac{\gamma(1-\rho+\rho q) -
    \rho q}{1-\rho+\rho q-\gamma\rho q},
\end{equation}
which we call the MAF coefficient ratio.

Similarly, the PCA coefficients are determined by maximizing total variance,

\begin{equation}
  \label{eq:7}
   S+N(w_1, w_2) = S(w_1, w_2) + N(w_1, w_2) = \max_{\nu}\left\{ (qw_1b_1+qw_2b_2)^2
     +(1-\rho) + \rho(q w_1+q w_2)^2 \right\},
\end{equation}
which is maximized when
\begin{equation}
  \label{eq:45}
   \nu_{PCA} = \frac{w_2}{w_1} =  \sqrt{\alpha^2 + 1} - \alpha,  \text{   with  }
  \alpha = \frac{b_1^2 - b_2^2}{2(b_1 b_2 + \rho)}.
\end{equation}
\end{lem}

\begin{proof}
For the MAF result, substitute the specific parameter values of the above model into
the general expression for SNR$(\bm w)$ in \eqref*{eq:R} to obtain \eqref*{eq:5}. Thereafter,
find the maximum of the quadratic expression in
\eqref*{eq:5}. Note that the minimum is attained when $\nu = -1/\gamma$. For the PCA result, find the values of $w_1$ and $w_2$ for which \eqref*{eq:7} is maximized under the constraint that $w_1^2 + w_2^2= 1$.
\end{proof}

Note that the input parameters investigated here are the cross-correlation in the noise, $\rho$, the relative differences in the two groups' SNR, $\gamma$, and the overall number of time series, $p = 2q$.

In particular, consider what happens to MAF and PCA SNR when changing the number of time
series in each group, $q$. Note that $\nu_{PCA}$ does not depend on $q$, while $\nu_{MAF}$
does. Taking limits in $q$,
\begin{equation}
  \label{eq:49}
  \text{SNR}(\bm w_{MAF}) \sim q\frac{b_1^2 (1 - \gamma)^2}{2(1-\rho)} \text{ as } q\rightarrow \infty.
\end{equation}

Similarly,
\begin{equation}
  \label{eq:1}
  \text{SNR}(\bm w_{PCA}) \sim \frac{b_1^2 (1
    + \nu_{PCA} \gamma)}{\rho (1 + \nu_{PCA})^2} \text{ as } q\rightarrow \infty.
\end{equation}
Thus, the associated SNR of PCA approaches a
constant, while SNR of MAF will grow linearly with $q$. This implies that the MAF SNR continues to improve as the number of time series increases, unlike PCA which reaches a plateau. Furthermore,  $
\lim_{q \to \infty}\nu_{MAF} = - 1$, a result that intuitively follows from the fact that the noise has equal variance across the groups, unlike the signal. Thus, if $w_1 = -w_2$ and $q$ is large enough the noise component will cancel while the signal remains.

In \fig{MAFPCAratio}, MAF and PCA SNR values are compared as
$\gamma$ and $\rho$ are changed. The ratios of the SNRs are plotted in
a contour plot. Each panel shows a different $p$, the
total number of time series. We see that increasing the
cross-correlation, $\rho$, increases the difference between MAF and
PCA SNR while increasing $\gamma$ has the opposite effect. Increasing
the number of time series will exacerbate the difference between the SNRs, as explained in the asymptotic analysis above.

\EPSFIG[scale=0.6]{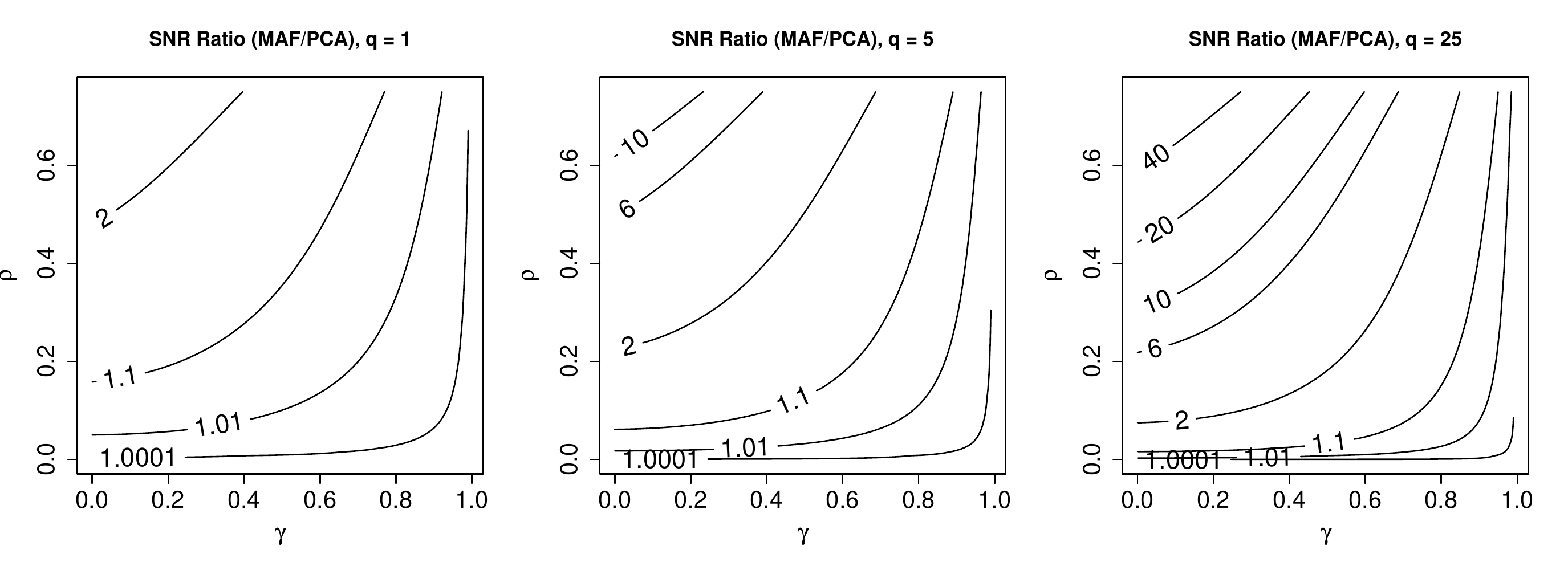}{The SNR of MAF divided by the SNR
  of PCA when $q = {1, 5, 25}$, with  $\gamma$ and $\rho$
  changing. }{MAFPCAratio}

\subsection{Model II: A model with common cross-correlation and different variances}

To generalize Model I, we allow each time series to have a unique signal strength and noise variance. The following lemma derives the form the MAF coefficient vector, $\bm w_{MAF}$, takes in this model.

\begin{lem}
Consider the multivariate time series model,
\begin{equation}
  \label{eq:22mod}
  \bm Z(t) = f(t) \bm b + \bm \varepsilon(t),
\end{equation}
where $f(t)$ is the signal time series, $b$ is the vector of signal
strengths for each time series, and $Cor[
\varepsilon_i(t), \varepsilon_j(t)] = \rho$ and $Var[\varepsilon_i(t)]
= \sigma_i^2$.

Then, the MAF coefficient vector $\bm w = (w_1, \dots, w_p)$, is given by
\begin{equation}
  \label{eq:6}
  w_i \propto \; \frac{b_i}{\sigma_i^2} -
    \frac{\rho}{1+\rho(p-1)}\sum_{j=1}^{p}\frac{b_j}{\sigma_i\sigma_j},
  \quad i = 1, 2, \dots, p.
\end{equation}
\end{lem}

\begin{proof} Express the noise covariance matrix as
\begin{equation}
  \label{eq:9}
  \bm \Sigma_\varepsilon = \bm A \left[\rho \bm 1_p \bm 1_p + (1-\rho) \bm
    I \right] \bm A
\end{equation}
where $A_{ii} = \sigma_i$ and $A_{ij} = 0$ for $i\neq j$. Thus,
\begin{equation}
  \label{eq:2}
  \bm\Sigma_\varepsilon^{-1} = \frac{1}{1-\rho}\left[ \bm A^{-2} -
    \frac{\rho}{1+\rho(p-1)}\bm A^{-1}\bm 1\bm 1'\bm A^{-1}\right].
\end{equation}
Now it can be shown that $\bm w_{MAF} =
\bm\Sigma_\varepsilon^{-1}\bm b$ (see \eqref*{eq:24} in Appendix \ref{app:proofs})
and the lemma follows by substitution.
\end{proof}

We now consider the special case where all input time series have the same noise
variance. Substitution into \eqref*{eq:6} gives the MAF coefficient
vector
\begin{equation}
  \label{eq:13}
   w_i \propto \; b_i -
    \frac{\rho}{1+\rho(p-1)}\sum_{j=1}^{p}b_j,
  \quad i = 1,2, \dots, p.
\end{equation}

An alternative way to derive the MAF coefficients in this special case of common noise variance for all time
series is to find the eigenvectors of \eqref*{eq:15} analytically. Using this method also illustrates that only in this special case of common noise variance can the leading MAF coefficient
vector in \eqref{eq:13} be constructed from a linear combination of
the first two PC coefficient vectors. However, when the noise variances are
not all equal, this will not be the case. In fact, no linear combination
of the PCs can be used to obtain the MAF time series. More details are given in Appendix \ref{app:coeffs}.

If furthermore, $\rho=0$, i.e. no cross-correlation between noise time
series, then the MAF and PCA coefficient vectors are the same and are
proportional to signal strength vector $\bm b$.


\section{The MAF methodology on sampled time series}
\label{sec:Methodology}

In the following section, we specify how the MAF methodology is implemented on a given collection of sampled time series. First, an algorithm which calculates the MAF factors is specified. Second, a set of methods for the selection of number of MAFs is discussed. Lastly, we investigate how to quantify the uncertainty in the estimated MAF factors in the S+N model.

\subsection{Calculation of Maximum Autocorrelation Factors}

Consider an input set of $p$ time series, each of which is recorded at $n$ time points. Let the column matrix $\bm Z \in \R^{n \times p}$ with $n > p$ denote this collection of time series. First recall from Equation \ref{eq:15} that the MAF factors are defined as the eigenvectors of $\bm S^{-1/2} \bm S_{\Delta} \bm S^{-1/2}$. The following operations are implemented on the data matrix $\bm Z$ to obtain these eigenvectors.

 First, we transform $\bm Z$ such that its covariance matrix, $\bm S$, is the identity through spectral decomposition. Working with this transformed matrix of new time series we compute the first differences in time and compute the corresponding covariance matrix, $\bm S_\Delta$. Then, we obtain the eigenvectors of the differenced covariance matrix via spectral decomposition. These eigenvectors are in turn transformed back to the original coordinate system to become the columns of the MAF coefficient matrix, defined as $\bm W_{MAF}$. Finally, $\bm W_{MAF}$ is pre-multiplied by $\bm Z$ to yield the MAF factors, and are the orthogonal time series with maximum autocorrelation. Algorithm \ref{alg:MAFvector} formally specifies the calculation of these MAF factors. For the purpose of this algorithm, the covariance operator is defined as $Cov(\bm Z) = \sum_{i=1}^{n} [\bm Z_{i\cdot} - \sum_{k=1}^n \bm Z_{k\cdot}][\bm Z_{i\cdot} - \sum_{k=1}^n \bm Z_{k\cdot}]'$, where $\bm Z_{i\cdot}$ is the $i^{th}$ row of $\bm Z$.

\begin{algorithm}
\caption{Calculate MAF factors, $\bm Y \in \R^{n\times p}$, and MAF coefficients $\bm W_{MAF}(\bm Z).$}
\label{alg:MAFvector}
\begin{algorithmic}[1]
\REQUIRE $\bm Z \in \R^{n \times p}$
  \STATE Calculate $\bm S_{\bm Z} = Cov(\bm Z)$
  \STATE Decompose $\bm S_Z$ such that $\bm S_Z = \bm U \bm D \bm U'$ where $\bm U \in O(p)$ and $D$ is diagonal with the eigenvalues of $\bm \Sigma_Z$.
  \STATE Compute $\bm X = \bm Z \bm U \bm D^{-0.5} \bm U'$.
  \STATE Compute $\Delta \bm X_{i\cdot} = \bm X_{i\cdot} - \bm X_{(i+1)\cdot}$
  \STATE Compute $\bm S_{\Delta} = Cov(\Delta\bm X)$.
  \STATE Decompose $\bm S_\delta = \bm V \bm K \bm V'$ where $\bm V\in O(p)$ and $\bm K$ is diagonal with eigenvalues in increasing order, $ K_{11} \leq K_{22} \leq \dots \leq K_{pp}$.
  \STATE Let $\bm W_{MAF}(\bm Z) = \bm U \bm D^{-0.5} \bm U' \bm V$
  \STATE Compute $\bm Y = \bm Z \bm W_{MAF}(\bm Z) $.
  \STATE Let $\bm Y_{\cdot j}$ be the $j^{th}$ column of $\bm Y$. For each $j$, compute $c_j = \text{sign} \left[\sum_{i=1}^n i \bm Y_{\cdot j}\right]$.
  \ENSURE $c_j \bm Y_{\cdot j}$ for each $j$ and $\bm W_{MAF}(\bm Z).$
\end{algorithmic}
\end{algorithm}

\subsection{Uncertainty quantification}
\label{sec:UncertaintyQuantification}

Often there is a need to understand the sampling variability of the estimated $\bm w_{MAF}(\bm Z)$ and the associated MAF factors in the S+N model. One natural way to undertake this is to use resampling of the data.

In resampling the time series we seek to preserve the underlying signal while resampling the noise. As such, an underlying smooth signal estimate is obtained by smoothing the original time series. Denote the smooth estimate as $\bm \tilde{Z_i}(t)$, for $i = 1, \dots, p$. Possible smoothing techniques include local regression (Loess) \citep{Cleveland1979} and spline smoothers \citep{ESL}.

The residuals between the original and the smooth time series, $\hat{\bm \varepsilon}(t) = \bm Z(t) - \tilde{\bm Z}(t)$, is then resampled and added back to the smooth original time series, $\tilde{\bm Z}(t)$. Resampling can be done in blocks as there is temporal structure in the residuals. Denote the resampled time series as $\bm Z^{\ast}(t)$. With the new set of time series we recompute the MAF coefficients and MAF factors, $\bm w_{MAF}(\bm Z^{\ast})$ and $\bm Y^\ast(t)$.

The above procedure can be repeated to obtain B instances of $Y_i^\ast(t)$ that can be aggregated to obtain a pointwise confidence interval around $Y_i(t)$. In the event that the resampled MAFs are not well centered around the original MAFs $\bm Y(t)$, information might have been lost as the noise vector was resampled, compromising the shape of the MAF in question. This would suggest that there is too much noise present for any isolation of a signal. Thus, a significance test could be employed to determine the MAF's relevance. The procedure to resample the MAF coefficients and MAF factors is formally given in Algorithm \ref{alg:MAFUncertainty}.

\begin{algorithm}
\caption{Resample the MAF factors, $\bm Y \in \R^{n\times p}$}
\label{alg:MAFUncertainty}
\begin{algorithmic}[1]
\REQUIRE $\bm Z \in \R^{n \times p}$
  \STATE For the set of time series $ Z_i(t)$ for $i = 1, ..., p$, calculate $\bm w_{MAF}(\bm Z)$ and $\bm Y(t)$.
  \STATE Create a smooth time series from each original time series $Z_i(t)$ and calculate the residual $\hat{\varepsilon}_i(t) = Z_i(t) - \tilde{Z}_i(t)$.
  \STATE From the set of integers $[1,n]$, draw a sequence $\{s_i\}$ of $n$ integers with replacement, yielding the sequence ${s_1}, {s_2}, \dots, {s_n}$. The new set of residuals becomes $\hat{\varepsilon}_i({s_1}), \hat{\varepsilon}_i({s_2}), \dots, \hat{\varepsilon}_i({s_n})$ which we shall call $\hat{\varepsilon}_i^{\ast}(t)$. This step could be modified to resample blocks of residuals.
  \STATE Let $Z_i^{\ast}(t) = \tilde{Z}_i(t) + \hat{\varepsilon}_i^{\ast}(t)$.
  \STATE Calculate new MAFs from the resampled data, $\bm Y^\ast(t) = \bm Z^\ast \bm w_{MAF}(\bm Z^{\ast})$.
  \STATE Repeat steps 3-5 B times.
  \ENSURE B realizations of $\bm Y^\ast(t)$ and $\bm w_{MAF}(\bm Z^{\ast})$.
\end{algorithmic}
\end{algorithm}

\subsection{Selection of number of MAFs}

In real applications, the number of underlying signals is often unknown. Determining the number of underlying signals can be done in various ways. First, one can find that the eigenvalues of $\bm S^{-1/2} \bm S_\Delta \bm S^{-1/2}$ and plot them in an ``autocorrelation scree plot'' similar to what is done in PCA. Using this plot, one can look for the presence of a shoulder to define the number of MAFs one should retain. Alternatively, one could define a cutoff after some fraction $\alpha$ (such as $95 \%$), of the total autocorrelation that is contained in preceding MAFs.

A second method employs cross validation to find the number of underlying signals. By defining a hold-out block one can regress each original time series in $\bm Z$ on the $k$ first MAFs for $k = 1, \dots, p$. Then, select $k$ such that the RMSE on the hold-out block is minimized.

 A third method involves using the framework of hypothesis testing, the description of which is deferred to the section on statistical inference in Subsection \ref{sec:inference}.

\

\section{Statistical properties}
\label{sec:samplingProperties}
Having looked at the model properties of MAF and PCA, we now turn to
their sampling properties. This section is divided into three parts. In the first
part, we show that the sample covariance and lagged covariance yield
consistent estimates of MAF coefficients
and MAF factors under the S+N model as the number of time steps grows. PCA
estimates are treated similarly. In the second part, a simulation study is conducted to compare MAF and PCA as signal
recovery techniques. We use the signal cross-correlation with MAF and PC time series as
the metric of comparison, and find that MAF is both more resilient to
increased noise and more suitable when the noise has
cross-correlation. In the third subsection, we introduce a hypothesis
testing framework to test if an underlying time trend extracted by MAF
is statistically significant.

\subsection{Consistency}

The following theorem shows that as the number of time steps grows for
a $p$-variate time series, the MAF and PCA coefficients will
converge to their model values.

\begin{thm}
\label{theorem1a}
Consider a set of time series $Z_n(t) \in \R^p$, such that
\begin{align}
  \label{eq:21mod}
  \bm Z_n(t) =& f_n(t) \bm b + \bm \varepsilon_n(t) \quad t= 1, \dots, n
  \nonumber \\
  \bm \Delta \bm Z_n(t) =& \bm Z_n(t) - \bm Z_n(t+1) =  \Delta f_n(t) \bm b + \Delta \bm \varepsilon_n(t)
\end{align}
with $f_n(t) \in \R \:\: \forall t = 1,2,...,n , \bm b, \bm \varepsilon_n(t)  \in \R^p\:\: \forall t = 1,2,...,n $, $\Delta
\bm \varepsilon_n = \bm \varepsilon_n(t) - \bm \varepsilon_n(t+1)$,
and $\Delta f_n(t) = f_n(t) - f_n(t+1)$. Residual time series $\bm\varepsilon_n$ is a
weakly stationary $p$-variate time series and the associated autocovariance is
absolutely summable. The
signal time series is such that
\begin{equation}
  \label{signalConditions}
\frac{1}{n}\sum_{t=1}^{n} f_n(t) = 0, \quad
\frac{1}{n}\sum_{t=1}^{n} f_n^2(t) = 1, \quad
\frac{1}{n-1}\sum_{t=1}^{n-1} [\Delta f_n(t) - \overline{\Delta f_n(t)}]^2 = a, \quad \forall n,
\end{equation}
where $\Delta f_n(t) = \frac{1}{n}(f_n(1) - f_n(n))$.

 Then,
 \begin{align}
 \bm S_{n} \overset{p}{\to}&  \bm\Sigma \text{  as }n \to \infty, \nonumber\\
\bm S^{-1/2} \bm S_\Delta \bm S^{-1/2} \overset{p}{\to}& \bm \Sigma^{-1/2}  \bm \Sigma_\Delta \bm \Sigma^{-1/2} \text{  as }n \to \infty \nonumber\\
\end{align}
\end{thm}

\begin{proof}
See Appendix \ref{app:proofs}.
\end{proof}

\subsection{Simulation study}
\label{sec:simulationStudy}

We now undertake a simulation study to compare the MAF and PCA proceedures as signal recovery techniques. First, we generate $100$ simulations of $p=3$ parallel time series of length
$n=150$, using the S+N model of Section
\ref{sec:snModel}, viz.
\begin{equation}
  \label{eq:16}
  \bm Z(t) = f(t) \bm b + \bm \varepsilon(t)
\end{equation}
where $f(t)$ is a specified underlying signal time series shown in
\fig{signal}. This time series is a rescaled and interpolated version
of the mean annual surface time series for the northern hemisphere for
the years 1850-2007, taken from \citet{Mann2008}. The vector $\bm b$ is the
$p$-vector of signal strengths. $\bm\varepsilon(t)$ is an $iid$
zero-mean Gaussian noise $p$-vector for $t = 1, \dots, n$. The $3 \times 3$ cross-covariance
matrix for $\bm\varepsilon$ has a unit diagonal and common value
$\rho$ in the off-diagonal entries.

\EPSFIG[scale=0.35]{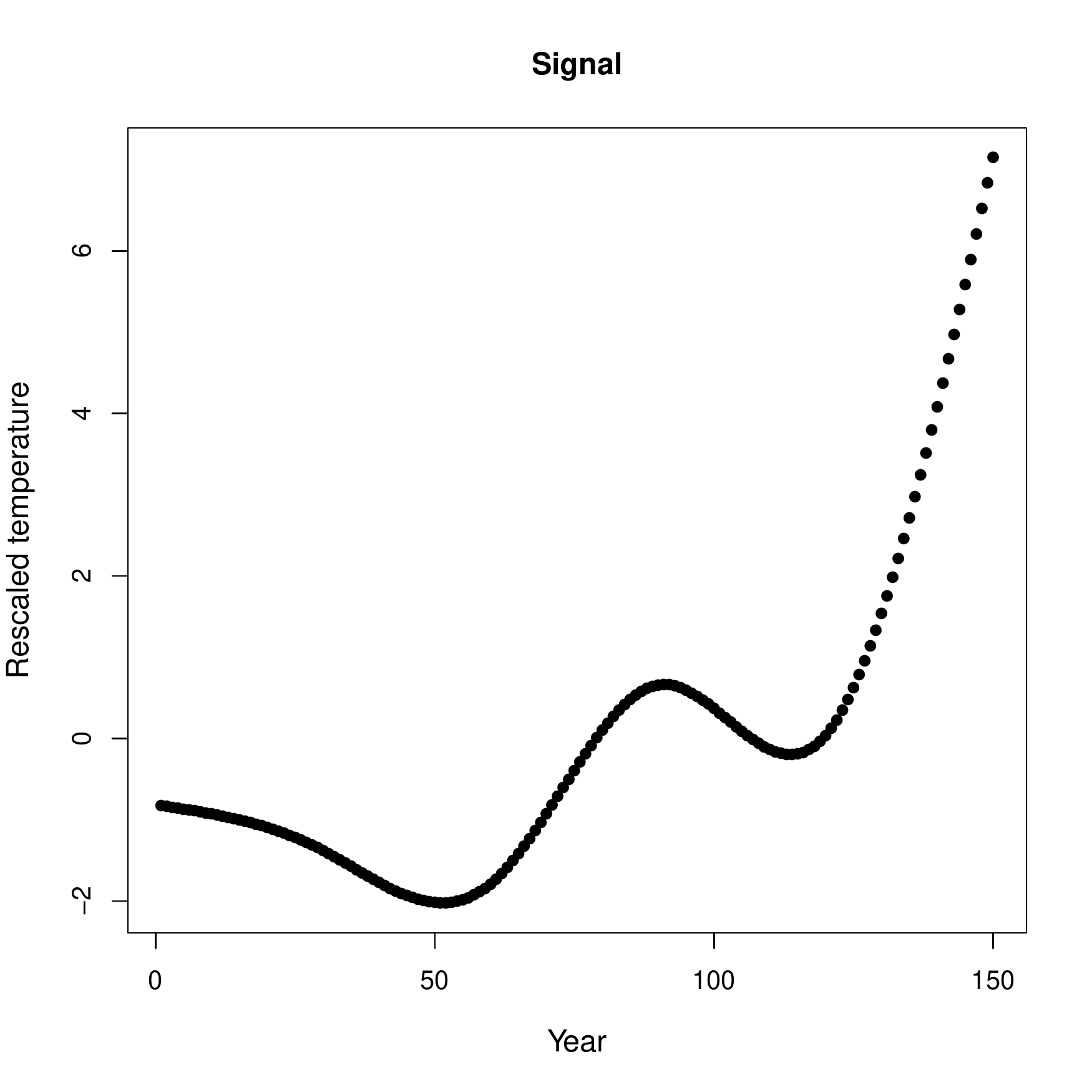}{Signal used in the signal recovery
  process, shown with a mean zero and unit variance.}{signal}

For each of the 100 realizations of the three parallel time series we
compute the combined MAF$(t)$ time series and the combined PCA$(t)$
time series. The cross-correlations of the MAF time series and the PCA
time series with the true signal time series as given in Figure \ref{signal} are used as a metric for
comparison.

Two specific simulations of the data are shown in \fig{ProxyChangingNoise}
with their associated smoothed MAF and PCA time series on the right. We use a LOESS filter to smooth the time series\footnote{We use local regression to smooth with 60 years in the span and tricubic weighting. The equivalent span as a fraction of the total time series is $60/150 = 2/5$ and the tricubic weight go as $(1- (d/60)^3)^3$ with $d$ the distance from the point of interest.}. The first
row of \fig{ProxyChangingNoise} shows the three parallel time series
with $\bm b= (0.8, 0.4, 0.2)$ and cross-correlation $\rho = 0.25$. The second row shows a
parallel time series with a weaker signal strength with $\bm b = (0.4,
0.2, 0.1)$ and $\rho = 0.25$. The results of the analysis are compelling. Namely, the cross-correlation of the MAF time series with the underlying
signal for the first row of \fig{ProxyChangingNoise} is $0.67$ and the
PCA equivalent is $0.52$, while for the second row MAF and PCA cross-correlations are $0.35$ and $0.08$, respectively.

\EPSFIG[scale=0.7]{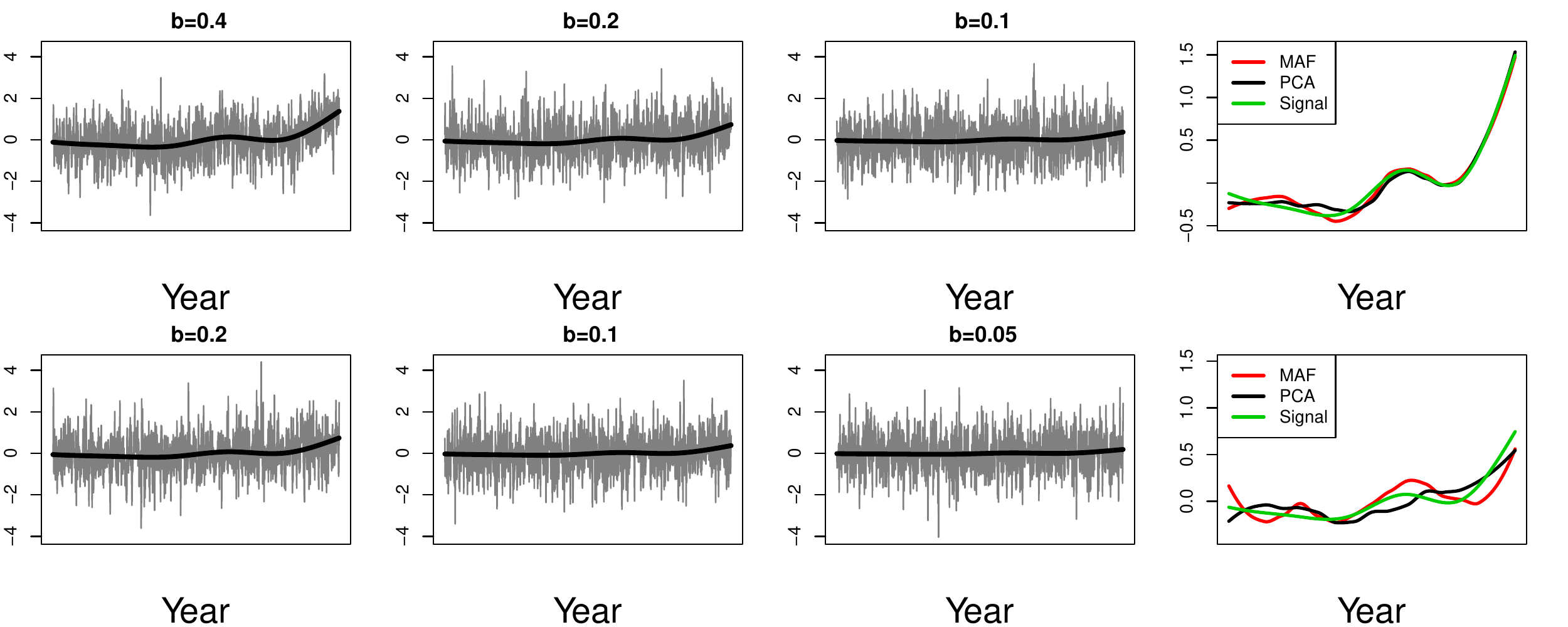}{Top: One realization of the
  data with the signal shown in bold on top of each time series. The
  Signal-to-Noise Ratio (SNR), corresponding to $\bm b=(0.8, 0.4, 0.2)$, is annotated above each figure and each time series has been scaled to have unit
  variance and zero mean. The smoothed MAF1 and PC1 are shown on the
  right. Bottom: Same as top with different SNR, $\bm b=(0.4, 0.2,
  0.1)$. The noise cross-correlation, $\rho = 0.25$.}{ProxyChangingNoise}

We proceed to undertake further analysis of this model in order to fully understand how MAF and PCA perform when the cross-correlation, $\rho$ and the signal strength vector $\bm b$ changes. The full set of scenarios that we consider in this example are:
\begin{itemize}
\item  A fixed signal strength vector, $\bm b=(0.8,0.4,0.2)$, with
  changing noise cross-correlation $\rho$.
\item A fixed noise cross-correlation $\rho=0.25$ with changing $\bm b =
  (0.8c,0.4c,0.2c)$ for $c\in[0.5,2.5]$.
\end{itemize}

\fig{CORCrossAndSNR} contains plots of signal cross-correlations with
MAF and PCA time series for each of the parameter
combination scenarios. Each plotted point represents an average over 100
simulations. MAF yields higher correlation with the signal uniformly. It is clear that MAF takes advantage of
cross-correlation in the noise and uses it to amplify the signal,
while PCA fails to exploit this property in the noise and thus
under-performs compared to MAF.

The signal information contained jointly across PC1 and PC2 is also less than the information contained in MAF1 only. This can be seen by regressing the signal of both PC1 and PC2 and extracting the root of the $R^2$ value. This is the multivariate equivalent to correlation between the a signal and a signal estimate. This result is also shown in \fig{CORCrossAndSNR} under the legend PC$1+2$.

\EPSFIG[scale=0.65]{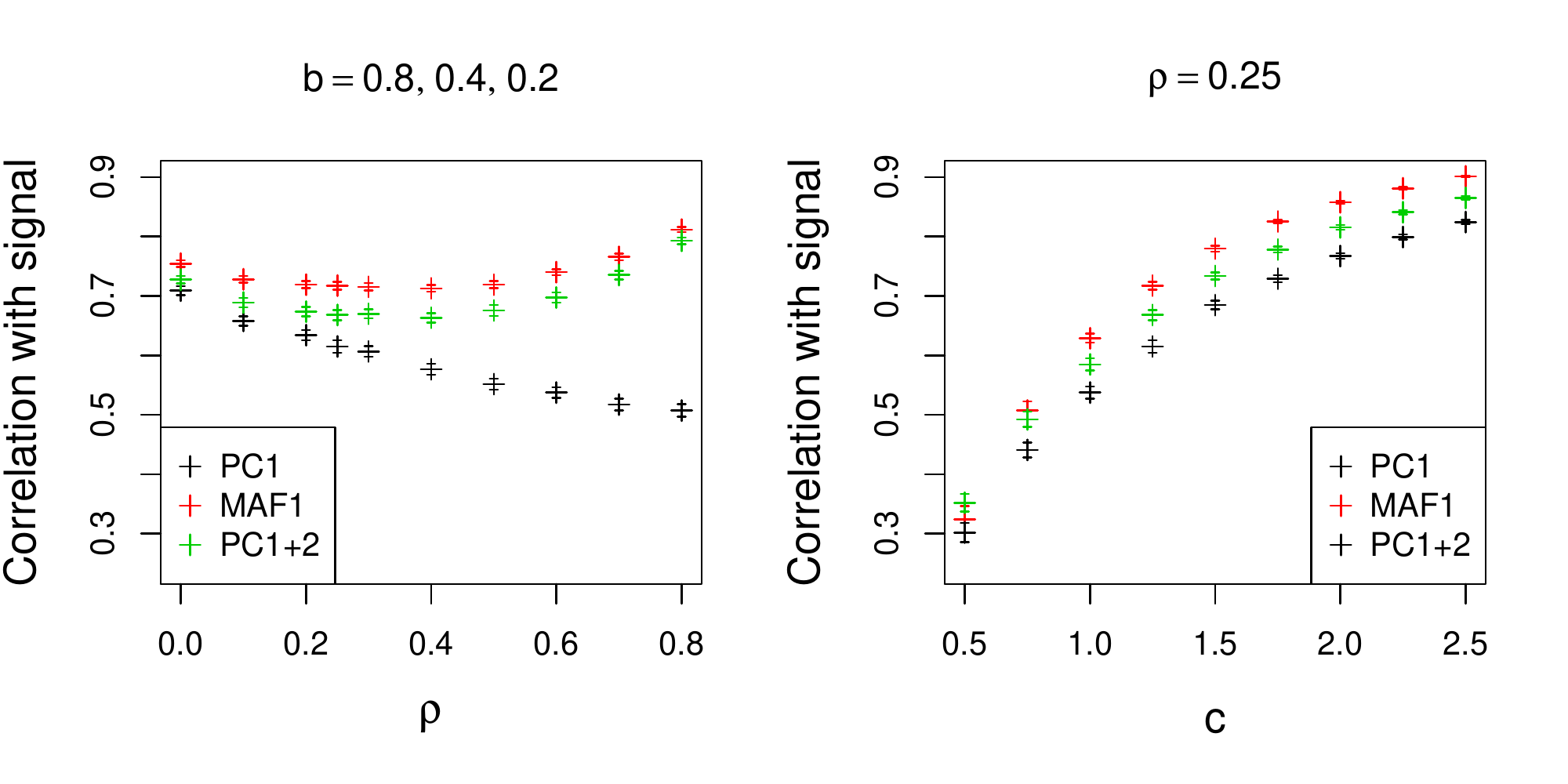}{Left: Correlation of signal
  estimate, using MAF or PCA, with true signal while changing cross-correlation
  of noise. Right: Correlation of signal
  estimate, using MAF or PCA, with true signal while multiplying the signal
  strength vector
  by a factor as shown on the x-axis. The error bars show twice the standard
  error from the mean using 100 repetitions. }{CORCrossAndSNR}

\subsection{Hypothesis Testing}
\label{sec:inference}

Consider the same $p$-variate time series of length $n$ as described in \eq{eq:16}.
One might want to test whether a time signal is indeed present in the
data or not. We consider the following hypotheses:
\begin{align}
\label{eq:hypTest}
H_0: &   \bm Z_n(t) =  \bm \varepsilon(t) \nonumber\\
H_A: & \bm Z_n(t) = f(t) \bm b + \bm \varepsilon(t), \quad f(t) \neq
 \text{constant},
\end{align}
where $\bm\varepsilon(t)$ is an $iid$ zero-mean Gaussian
noise p-vector time series with cross-correlation $\rho$ and unit variance.

To test for the presence of a signal in a MAF we introduce empirical signal-to-noise ratio, $SNR_{empir}(x(t))$, as a function of an arbitrary time series $x(t)$ for $t=1, ..., n$, using a smooth version of $x(t)$, called $\tilde{x}(t)$,
\begin{equation}
   SNR_{empir}(x(t)) = \frac{SD(\tilde{x})}{SD(x - \tilde{x})},
 \end{equation}
surpressing the $t$ argument for brevity and with $SD = \sqrt{\sum_{t=1}^n\left[x(t) - \sum_{k=1}^n x(k)\right]^2}$.

Under the null model, there is no signal. So, subtracting out a smooth trend should not affect the corresponding null distribution. The only difference would be the slightly reduced degrees of freedom of the $\chi^2$ associated with the residuals after regressing on the smooth. This is accounted for by inflating the residuals by a factor of $\frac{n}{n-\nu_i}$ where $\nu_i$ is the degrees of freedom associated with each smoothed time series. We then resample the resulting residuals recompute the test statistic. Resampling can be done in blocks if there is temporal structure in the residuals.

From the newly created time series, we can obtain new MAF factors. Through the test statistic, we see whether the first MAF factor remain after the noise estimate has been shuffled, an operation that should preserve the MAF factor if it indeed represents a signal.

Algorithm \ref{alg:hypTestNrMafs} provides the details to the hypothesis testing procedure.
\begin{algorithm}[H]
\caption{Calculate number of MAF factors, $\bm Y \in \R^{n\times p}$}
\label{alg:hypTestNrMafs}
\begin{algorithmic}[1]
\REQUIRE $\bm Z \in \R^{n \times p}$
  \STATE For each time series $ Z_i(t)$ for $i = 1, ..., p$, calculate $SNR_{empir}(Z_i(t))$.
  \STATE Create a smooth time series from each original time series $Z_i(t)$ and calculate the residual $\hat{\varepsilon}_i(t) = Z_i(t) - \tilde{Z}_i(t)$.
  \STATE From the set of integers $[1,n]$, draw a sequence $\{s_i\}$ of $n$ integers with replacement, yielding the sequence ${s_1}, {s_2}, \dots, {s_n}$. The new set of residuals becomes $\hat{\varepsilon}_i({s_1}), \hat{\varepsilon}_i({s_2}), \dots, \hat{\varepsilon}_i({s_n})$ which we shall call $\hat{\varepsilon}_i^{\ast}(t)$.
  \STATE Let $Z_i^{\ast}(t) = \sqrt{\frac{n}{n-\nu_i}} \hat{\varepsilon}_i^{\ast}(t)$.
  \STATE Calculate new MAFs from the resampled data, $\bm Y^\ast(t) = \bm Z^\ast \bm w_{MAF}(\bm Z^{\ast})$, and their associated $SNR_{empir}(\bm Y_i^\ast(t))$ for $i=1, ..., p$.
  \STATE Repeat steps 3-5 B times.
  \STATE Calculate MAF1's associated p-value,
  \begin{equation}
    p = \frac{
  \{\#\text{ of } SNR_{empir}(\bm Y_1^\ast(t)) > SNR_{empir}(\bm Y_1(t))\}}{B}
  \end{equation}
\end{algorithmic}
\end{algorithm}

The empirical SNR is used as test statistic since it's model counterpart maximizes the expected likelihood under the mode described in Equation \ref{eq:hypTest2} and proved in Lemma \ref{lem:expectedLikelihood}. Sample autocorrelation was also explored as a test statistic but was found to be less powerful than empirical SNR.

In resampling the residuals one can sample with or without replacement, where the former is referred to as the bootstrap. Permuting the residuals, i.e. resampling without replacement allows for exact type 1 error control
because we sample from the population as opposed to an estimate of the
population which is the case for the bootstrap. Furthermore, the validity of the bootstrap depends on the
empirical distribution's asymptotic convergence to the population
distribution, but the permutation test does not have this
requirement.

However, if there is autocorrelation present in the residuals, one would normally the data in blocks to account for the temporal structure. In this case, permuting the data is less suitable due to the smaller number of permutations possible. Sampling with replacement does not have a reduction in the number of possible combinations and might thus be there method of choice.

This hypothesis testing procedure can be extended to multiple signals. Because MAF solves an eigenvalue/eigenvector problem the MAF factors are orthogonal. As a corollary, the second MAF maximizes autocorrelation on a dataset that lies in the space perpendicular to the first MAF. Similarly, the third lies in the space perpendicular to the first two MAFs. In this vein, each MAF will produce a signal estimate orthogonal to the other MAFs. In the hypothesis testing framework, one would test whether each signal estimate is significant or not. Our method for creating the null distribution outlined in the single-signal case would still be valid. The only difference would be that multiple signals are subtracted out of the dataset followed by a permutation of the residuals.  To reflect the multiple signal extension in Algorithm \ref{alg:hypTestNrMafs}, one would only replace Step 7 by the calculation of the empirical SNR and p-values of MAFs 1 through $k$ for $k=1,...,p$.

We continue using the two examples presented in
\fig{ProxyChangingNoise}. The top panels show three time series where the
signal strength vector $\bm b_A = (0.4, 0.2, 0.1)$, while the lower panels
contain a weaker signal strength, $\bm b_B= (0.4, 0.2,
0.1)/2$. Furthermore, the
smooth line in each panel is the underlying signal before the noise is
added, while the gray lines show the raw observations used to
calculate the MAF transformation.

 The MAF SNR distributions are shown
in \fig{inferenceResults}, where the solid vertical
  line is the SNR of the original observations while the
  histograms represent the SNR of 1000 sets of permuted
  observations. The $p$-value represents the probability of the
  observed MAF SNR under the null hypothesis, i.e. the
  absence of a signal. In the
  strong-signal case the $p$-value is 0, corresponding to the left
  panel, while the weak-signal case has a $p$-value of 0.893,
  corresponding to the right panel of \fig{inferenceResults}.

\EPSFIG[scale=0.45]{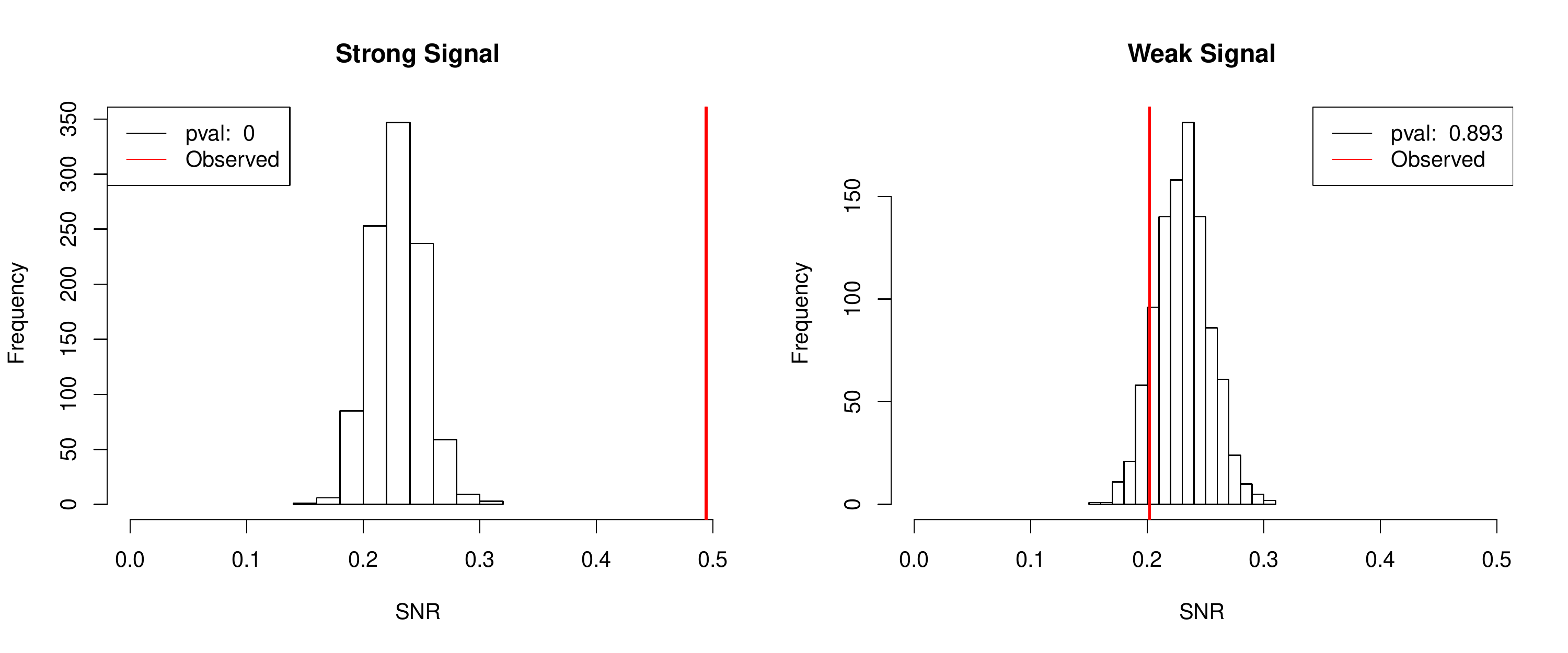}{The distributions of the maximized
  SNR under the null hypothesis, where the original data
  has been resampled by permuting the time steps. The black vertical
  line shows the original SNR, with the strong-signal example on the left
  and the weak-signal example on the right.  }{inferenceResults}

We proceed by calculating the power of the test under various signal
strengths. To
calculate the power we do the following,
\begin{enumerate}
\item Simulate B instances from the null model with no signal and calculate the associated test statistics.
\item Find the $(1-\alpha)^{th}$ quantile of the null distribution and call it $T_{1-\alpha}$.
\item Simulate B instances of the alternative and calculate the associated test statistics.
\item Find the area under the curve of the alternative distribution for which the test statistics are greater than $T_{1-\alpha}$. This area is the power.
\item Repeat steps 3-4 for different signal strengths.
\end{enumerate}

A plot of the power as a function of the power for a
number of signal strengths is shown in \fig{powerVsNoise}, with $B =
5000$, and $\rho = 0.5$. The x-axis represents the coefficient by which the base signal vector, $\bm b=(0.8, 0.4, 0.2)$, is multiplied. Both SNR and autocorrelation are used as test statistics and shown in separate panels. Note that SNR has higher power than autocorrelation.

\EPSFIG[scale=0.55]{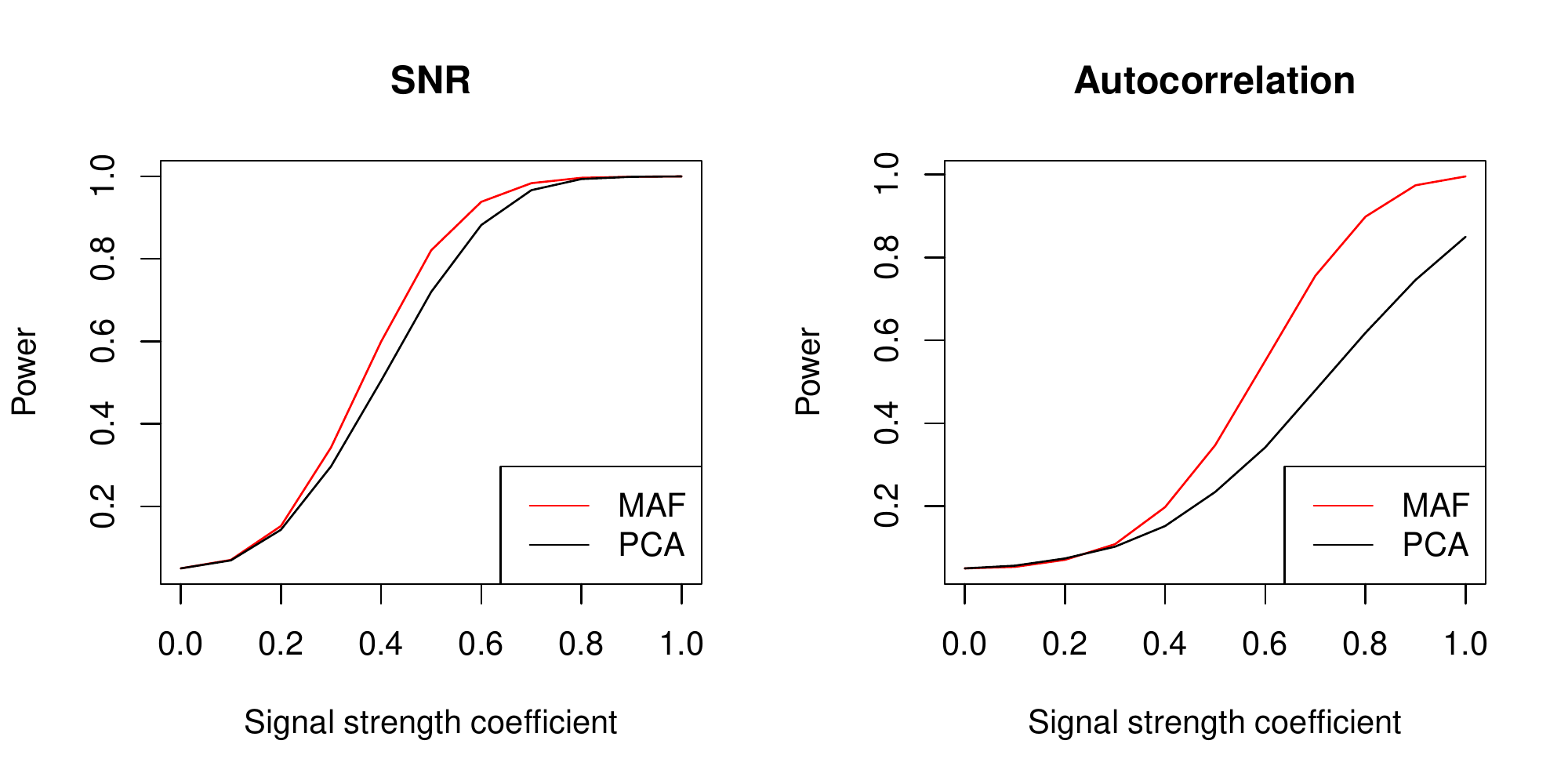}{The power of the hypothesis test as a
  function of the
  signal-multiplication factor. The base signal strength
  vector $\bm b = (0.8, 0.4, 0.2)$ is multiplied by $c \in [0, 1]$.}{powerVsNoise}


\section{Real Data: Application to Tree rings}
\label{sec:applications}
 To illustrate the efficacy of the MAF methodology, an application using the tree ring data from Western United States is presented in this section. First, we extract the MAFs to obtain an estimate of the underlying signal(s) present in the data. Uncertainty of the MAF factors is then estimated. Lastly, we test for the significance of the underlying signals through the hypothesis testing framework introduced in previous section.

\textbf{Overview of the data:} The data is obtained from the \citet{Mann2008}
and quantifies the annual growth of tree rings. It has been
pre-processed as described in Mann \emph{et al}'s Supplemental
Section. We selected 21 concurrent tree ring time series for the
period 1850-1999, of which 4 were already shown in
\fig{simpleExample}. \fig{TreeringFromCluster1} shows all 21 time series, scaled,
centered, and annotated by their
names\footnote{The raw data was download from the Supplemental section
from Mann at
\url{http://www.meteo.psu.edu/holocene/public_html/supplements/MultiproxyMeans07/}}. Some
time series show more temporal coherence than others. The goal here is to
extract a common underlying temporal signal.

\begin{sidewaysfigure}
    \includegraphics{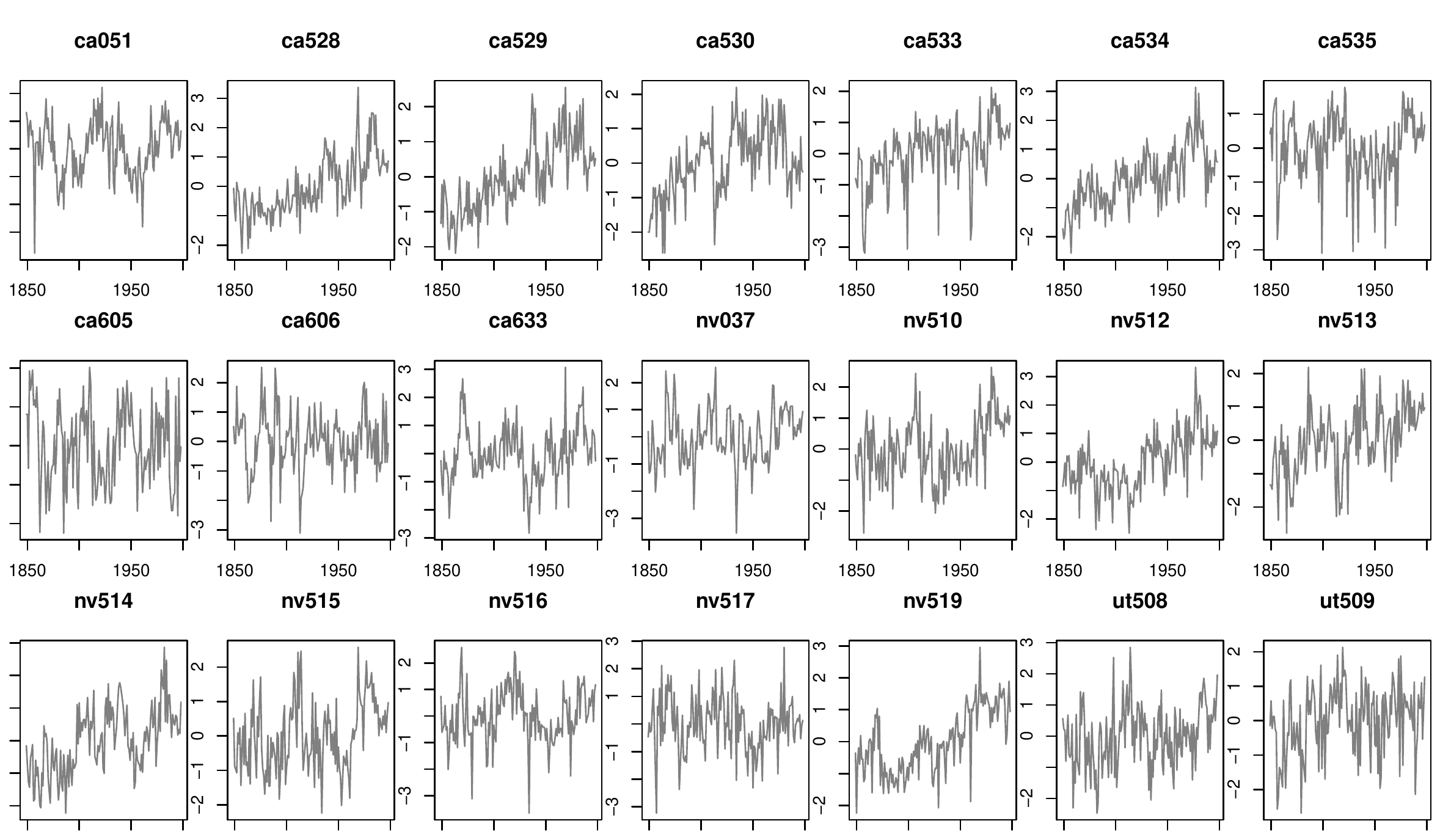}
    \caption{21 tree ring time series from the W77estern United States,
      centered and scaled.}
    \label{TreeringFromCluster1}
\end{sidewaysfigure}

\textbf{MAF estimation:} The 3 first MAFs and PCs are shown in \fig{mafpcaFromCluster1} where
each time series is annotated by its sample autocorrelation. A smooth
version of each time series is shown in bold with 30 years per
knot starting at the last year. Note that MAF produces time series that are more autocorrelated than PCA and are sorted in decreasing autocorrelation.

\EPSFIG[scale=0.55]{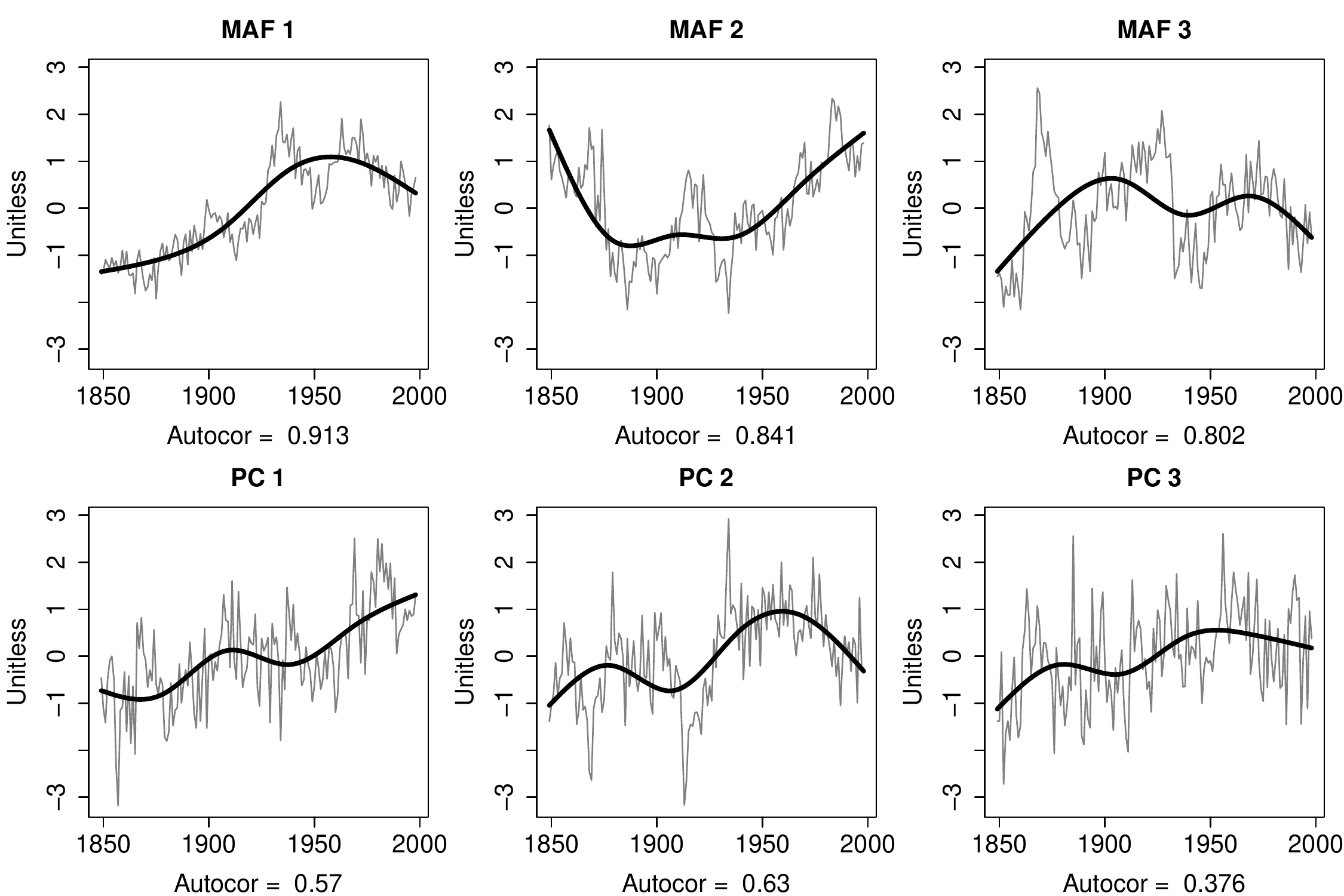}{Centered and scaled MAFs and PCs of the
  tree rings with smoothed equivalents shown in
  bold. }{mafpcaFromCluster1}

\textbf{Uncertainty quantification:} Quantifying uncertainty of the estimated MAF factors can be obtained
through the methodology outlined in subsection \ref{sec:UncertaintyQuantification}. A plot of this is shown in
\fig{TreeringUncertainty}, where 1000 resampled datasets are created by
doing a block bootstrap with a block size of 5 years.

Each new MAF
is created by using normalized MAF coefficients. The smoother
applied is the same LOESS smoother as in Section \ref{sec:simulationStudy}. We see a clear signal
present in the first two MAFs with the confidence bands containing the
original smooth MAFs. However, the third MAF time series' (MAF3) original estimate can
be seen almost outside the confidence interval. This suggests that MAF3 is mainly composed of noise such that when the tree
ring data is resampled and the MAF is recalculated the trend
associated with MAF3 disappears.

\textbf{Signal concentration:} The signal information also seems to be more focused in fewer MAFs compared with PC. This can be illustrated by obtaining the canonical correlation between the data set time series and the smooth MAFs and PCs. The first canonical correlation gives the linear combination of the data time series most correlated with a linear combination of MAFs/PCs. This can be interpreted as the correlation with a potential underlying signal.

\fig{canonicalMafPca} shows a sampling distribution of the canonical correlations for each data set previously obtained through resampling and the associated smooth MAFs and PCs. We see that MAF has a consistently higher canonical correlation until we include three components, at which point the two methods equalize. The result is even more pronounced using unsmoothed MAFs/PCs. Notice also that the MAF distributions are narrower than those of PCA. This is due to the smaller uncertainty about the MAF factor estimates shown in \fig{TreeringUncertainty}.

\EPSFIG[scale=0.55]{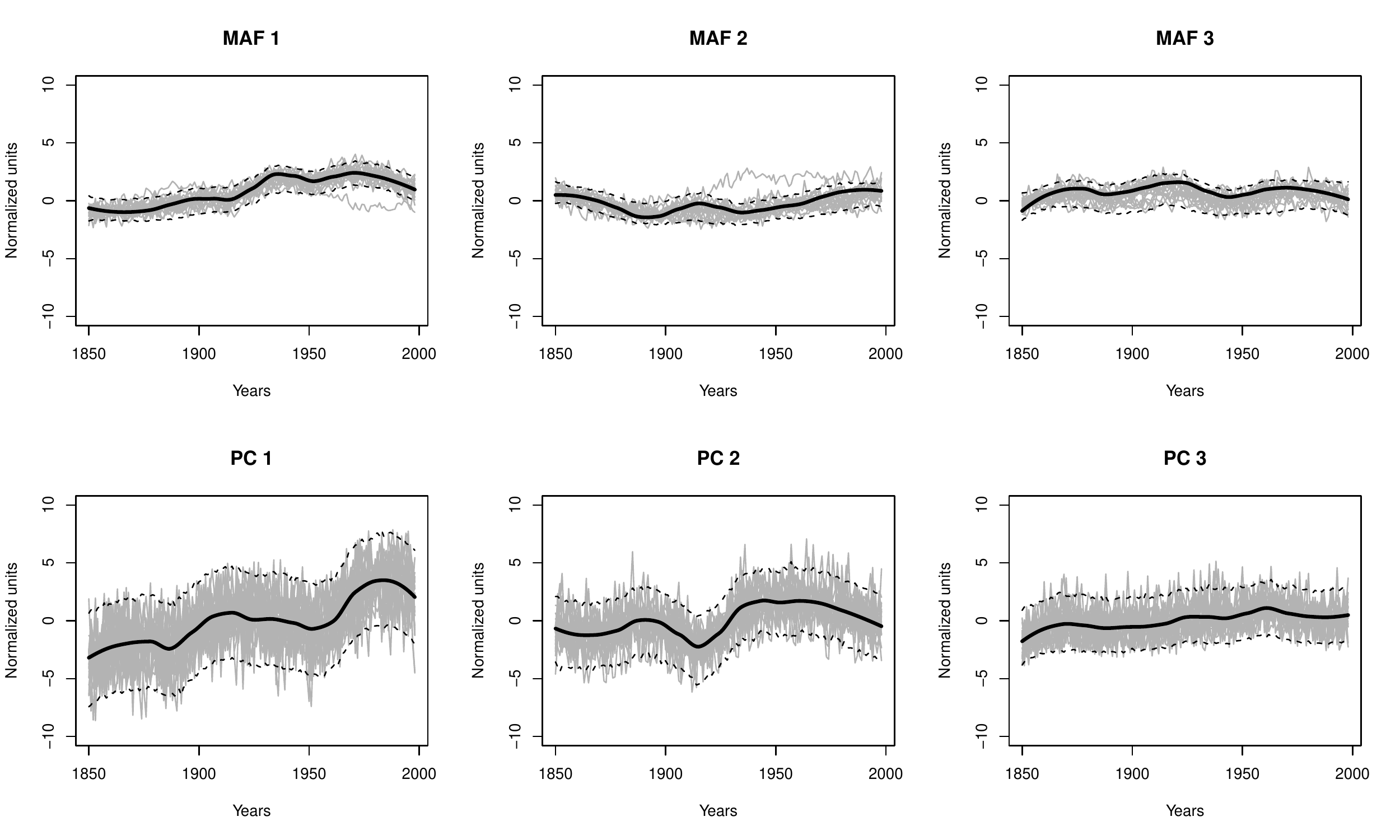}{A set of 1000 resampled MAF
  factors using the block bootstrap. The thick green line shows the original
MAF with a smoother applied, while the grey lines are the un-smoothed
resampled MAFs. Dashed lines show the $95^{th}$ confidence
bands. }{TreeringUncertainty}

\EPSFIG[scale=0.55]{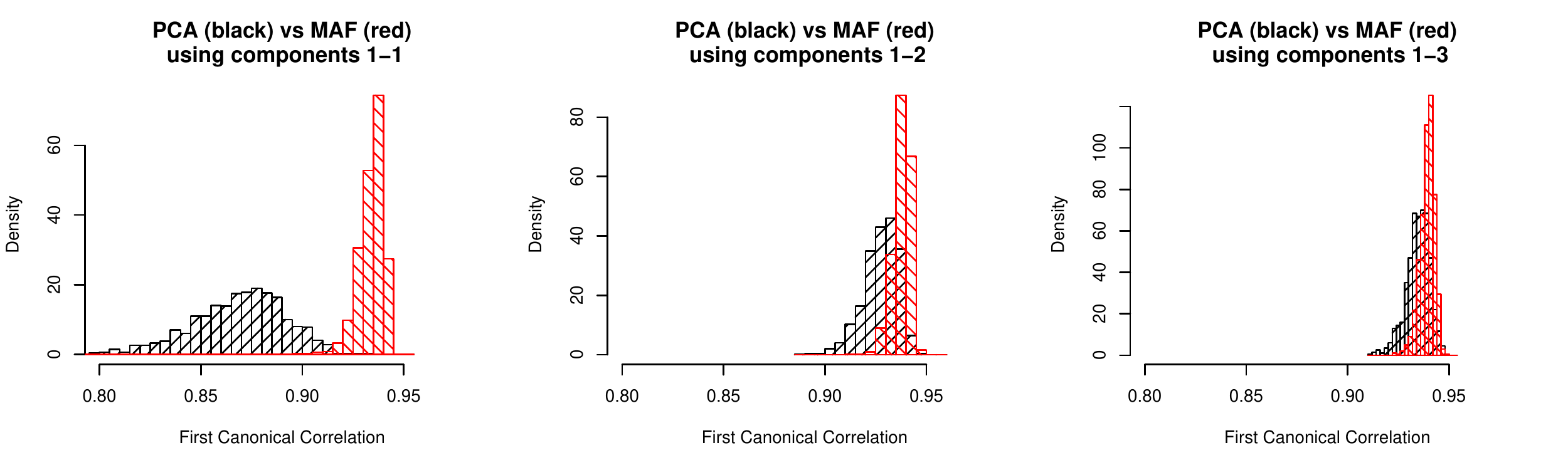}{The first canonical correlation between the resampled datasets and the associated MAF and PCA components 1 to $k$ for $k=1,2,3$.}{canonicalMafPca}


\section{Discussion}
\label{sec:conclusion}
We demonstrated advantages of the MAF optimization criterion in comparison
with PCA for the purpose of extracting a common time trend component
from multiple concurrent time series. In particular, under a model
where each time series is a combination of the underlying time trend
with additive noise, we showed that the MAF-optimized linear
combination of time series, i.e., maximizing autocorrelation, also
maximizes the signal-to-noise ratio among all possible linear
combinations. The sub-optimality of PCA can become worse as the number
of available time series grows, as the cross-correlation between time
series increases, and as the noise levels increase. We also
investigated some sampling properties of the MAF analysis and showed
through simulations that the MAF-optimized combined time series can be
statistically more stable than the corresponding PCA-optimized time
series obtained from the same set of concurrent time series data.

We generalize the signal-plus-noise model to include $q$ multiple underlying signal time series embedded in $p \geq q$ time series. Considering a response variable comprised of linear combinations of these predictive signals, we showed that the first $q$ MAFs span the same space as the first $q$ CCFs. And since CCA by definition maximizes correlation with the signals, the corresponding $q$ dimensional subspace spanned by the first $q$ CCFs is optimal when regressing the signal onto these. So, by transitivity, the first $q$ MAFs then also contain the optimal subspace for replicating the response. The advantage of MAF is that knowledge of the shape of the signal is not necessary. So, MAF compresses a $p$-dimensional time series into a $q$ dimensional data set without losing any information.

Lastly, we illustrated some initial applications of MAF applied to combining
21 concurrent annual tree ring time series for a region in western
North America, covering the period 1850-1999, with the goal of
extracting common time trend information. Regional tree ring time
series data are believed to be imperfect proxies for regional weather
time series, such as average annual temperature, and in a subsequent
paper we are investigating the calibration between regional
temperature time series and regional tree ring proxy time series for
these and other regions of the globe. An important step in the
calibration is the extraction of common time trend information from
the proxy data.


\bibliographystyle{plainnat}
\bibliography{MAF}

\begin{thebibliography}{15}
\providecommand{\natexlab}[1]{#1}
\providecommand{\url}[1]{\texttt{#1}}
\expandafter\ifx\csname urlstyle\endcsname\relax
  \providecommand{\doi}[1]{doi: #1}\else
  \providecommand{\doi}{doi: \begingroup \urlstyle{rm}\Url}\fi

\bibitem[Arbenz and Golub(1988)]{Arbenz1988a}
Peter Arbenz and Gene~H Golub.
\newblock On the spectral decomposition of hermitian matrices modified by low
  rank perturbations with applications.
\newblock \emph{SIAM Journal on Matrix Analysis and Applications}, 9\penalty0
  (1):\penalty0 40--58, 1988.

\bibitem[Briffa et~al.(2008)Briffa, Shishov, Melvin, Vaganov, Grudd,
  Hantemirov, Eronen, and Naurzbaev]{Briffa2008}
Keith~R Briffa, Vladimir~V Shishov, Thomas~M Melvin, Eugene~A Vaganov,
  H{\aa}ken Grudd, Rashit~M Hantemirov, Matti Eronen, and Muktar~M Naurzbaev.
\newblock Trends in recent temperature and radial tree growth spanning 2000
  years across northwest eurasia.
\newblock \emph{Philosophical Transactions of the Royal Society B: Biological
  Sciences}, 363\penalty0 (1501):\penalty0 2269--2282, 2008.
\newblock \doi{10.1098/rstb.2007.2199}.

\bibitem[Bunch et~al.(1978)Bunch, Nielsen, and Sorensen]{Bunch1978}
R.~Bunch, James, P.~Nielsen, Christopher, and C.~Sorensen, Danny.
\newblock Rank-one modification of the symmetric eigenproblem.
\newblock \emph{Numerische Mathematik}, 31\penalty0 (1):\penalty0 31--48--,
  1978.
\newblock ISSN 0029-599X.
\newblock URL \url{http://dx.doi.org/10.1007/BF01396012}.

\bibitem[Cleveland(1979)]{Cleveland1979}
W.S. Cleveland.
\newblock Robust locally weighted regression and smoothing scatterplots.
\newblock \emph{J. Amer. Statist. Assoc.}, 74\penalty0 (368):\penalty0
  829--836, 1979.

\bibitem[Doob(1953)]{Doob1953}
Joseph~L Doob.
\newblock \emph{Stochastic processes}, volume 101.
\newblock New York Wiley, 1953.

\bibitem[Durrett(2010)]{Durrett2010}
Rick Durrett.
\newblock \emph{Probability: theory and examples}.
\newblock Cambridge university press, 2010.

\bibitem[Gallagher et~al.(2014)Gallagher, Shaver, Bishop, Roginski, and
  Wise]{Gallagher2014}
Neal~B. Gallagher, Jeremy~M. Shaver, Randall Bishop, Robert~T. Roginski, and
  Barry~M. Wise.
\newblock Decompositions using maximum signal factors.
\newblock \emph{J. Chemometrics}, 28\penalty0 (8):\penalty0 663--671, August
  2014.
\newblock ISSN 1099-128X.
\newblock URL \url{http://dx.doi.org/10.1002/cem.2634}.

\bibitem[Hastie et~al.(2009)Hastie, Tibshirani, and Friedman]{ESL}
T.~Hastie, R.~Tibshirani, and J.~Friedman.
\newblock \emph{Elements of Statistical Learning}.
\newblock Springer, 2nd edition, 2009.

\bibitem[Jansen and Rajaratnam(2014)]{JB2012}
L.~Jansen and B.~Rajaratnam.
\newblock Robust reconstructions with temporal dependencies.
\newblock \emph{Journal of the American Statistical Association (in press)},
  2014.

\bibitem[Li et~al.(2007)Li, Nychka, and Ammann]{li07}
B.~Li, D.~W. Nychka, and C.M. Ammann.
\newblock The 'hockey stick' and the 1990s: a statistical perspective on
  reconstructing hemispheric teperatures.
\newblock \emph{Tellus A}, 59\penalty0 (5):\penalty0 591--598, 2007.

\bibitem[Mann et~al.(2008)Mann, Zhang, Hughes, Bradley, Miller, Rutherford, and
  Ni]{Mann2008}
Michael~E. Mann, Zhihua Zhang, Malcolm~K. Hughes, Raymond~S. Bradley, Sonya~K.
  Miller, Scott Rutherford, and Fenbiao Ni.
\newblock Proxy-based reconstructions of hemispheric and global surface
  temperature variations over the past two millennia.
\newblock \emph{Proceedings of the National Academy of Sciences}, 2008.
\newblock \doi{10.1073/pnas.0805721105}.

\bibitem[McShane and Wyner(2011)]{Mcshane}
B.B. McShane and J.~Wyner.
\newblock A statistical analysis of multiple temperature proxies: Are
  reconstructions of surface temperatures over the last 1000 years reliable?
\newblock \emph{The Annals of Applied Statistics}, 2011.

\bibitem[Muirhead(2005)]{muir}
R.J. Muirhead.
\newblock \emph{Aspects of Multivariate Statistical Theory}.
\newblock Wiley, 2005.

\bibitem[Shapiro and Switzer(1989)]{switzer89}
D.E. Shapiro and P.~Switzer.
\newblock Extracting time trends from multiple monitoring sites.
\newblock Technical report, Stanford University, 1989.

\bibitem[Switzer and Green(1984)]{switzer84}
P.~Switzer and A.~A. Green.
\newblock Min/max autocorrelation factors for multivariate spatial imagery.
\newblock Technical report, Stanford University, 1984.

\end{thebibliography}

\newpage
\section*{Supplemental section}
\appendix
\section{Getting general MAF coefficients}
\label{app:coeffs}

We present an alternative method for deriving the MAF coefficients under the
general model given in \eq{eq:22mod}. To do this, we first develop the case where all input time series have the same noise
level we get the following covariance for $\bm Z(t)$,
\begin{equation}
  \label{eq:34}
  \Sigma_Z = \bm b \bm b' + \rho \bm 1_p\bm 1_p' + (1-\rho)\bm I
\end{equation}
where $\rho$ is the common cross-correlation across all the time
series. The lagged covariance structure is given in \ref{eq3}. The PCs are given by the eigenvectors of
$\bm\Sigma_Z$, whereas the MAFs are the eigenvectors of
$\bm\Sigma_Z^{-1/2}\bm\Sigma_{\Delta Z}\bm\Sigma_Z^{-1/2}$.

First consider the special case where $\bar b = \sum_{i=1}^pb_i =
0$. This implies that $\bm b' \bm 1_p = 0$ and the eigenvectors are
$\bm b$, $\bm 1_p$ and all the vectors perpendicular to these two. The
corresponding eigenvalues are $(|\bm b|^2, \rho p, 0) + 1-\rho$
where the last eigenvalue is repeated $p-2$ times. If $|\bm b|^2 > \rho
p$, PC1 will be $\bm b$. PC2 will then be $\bm 1_p$, unless $\rho<0$
in which case PC2 will be in the aforementioned nullspace. Lastly, if
$|\bm b|^2 < \rho p$, PC1 will be $\bm 1_p$.

Another special case if where $\rho=0$. Here
the highest eigenvalue and corresponding eigenvector will be
proportional to $\bm b$ while all the others will be perpendicular to
$\bm b$ for both MAF and PCA.

In the general case where $\bar b \neq 0$ and $\rho \neq 0$, notice that $\bm b \bm b' + \rho \bm 1_p \bm 1_p'$ is rank 2. So the
dimensionality of that nullspace is $p-2$. All vectors in this nullspace
will have an eigenvalue of $1- \rho$. The remaining two eigenvectors
are found by assuming a
general structure of the eigenvectors, $\bm v = a_1 \bm 1_p + a_2 \bm b$.
It then follows that
the eigenvectors/eigenvalues of $\Sigma_Z$ are given by
\begin{align}
  \label{eq:35}
  v_1 =& \left(\frac{\rho p - |\bm b|^2 + \Delta}{2\bar{\bm b}p}\right)\bm 1_p + \bm b, \quad
  & \lambda_1 = \frac{\rho p + |\bm b|^2 + \Delta}{2} + 1 - \rho
  \nonumber\\
  v_2 =& \left(\frac{\rho p - |\bm b|^2 - \Delta}{2\bar{\bm b}p}\right)\bm 1_p + \bm b, \quad
  & \lambda_2 = \frac{\rho p + |\bm b|^2 - \Delta}{2} + 1 - \rho
\end{align}
where
\begin{align}
  \label{eq:36}
\Delta^2 =& (|\bm b|^2 - \rho p)^2 + 4(\bar{\bm b}p)^2\rho,
\end{align}
where $|\bm b|^2$ is the squared sum of the SNRs of the input time series.

The vectors in the nullspace have mean equal to
zero. This can be seen by considering any eigenvector $\bm v \in
\mathcal{N}(\bm b \bm b' + \rho \bm 1_p \bm 1_p')$,
\begin{align}
  \label{eq:nullspace}
(\bm b \bm b' + \rho \bm 1_p \bm 1_p')\bm v = 0 \nonumber \\
(\bm b'\bm v) \bm b + \rho \bar{v}p \bm 1_p = 0.
\end{align}
Because $\bm b$ is in general not equal to $\bm 1_p$, we have $\bm
b'\bm v = \bar{v} = 0$

To get the eigenvectors corresponding to the MAFs, consider also the
lagged covariance matrix,
\begin{align}
  \label{eq:39}
  \Sigma_{\Delta Z} =& k_f \bm b \bm b' + k_{\varepsilon} (\rho \bm 1_p \bm 1_p'
  + (1-\rho)\bm I)
\end{align}

Letting $\bm \Sigma_Z^{-1} = \bm\Gamma \bm D^{-1}\bm\Gamma ' = \bm H
\bm H'$, where $\bm H = \bm\Gamma
\bm D^{-1/2}$. Furthermore, because the optimal SNR in \eq{eq:R}
does not depend on $k_f$ as long as $k_f < k_{\varepsilon}$, we can
set $k_f = 0$, w.l.o.g. This is because the optimal SNR coefficients
for each time series is equivalent to the MAF1 coefficients, which is the
eigenvector corresponding to the smallest eigenvalue of
\begin{equation}
  \label{eq:orgiCov}
  \tilde{\bm \Sigma}_{\Delta} = \bm H'  \left(\rho \bm 1_p \bm 1_p' +
    (1-\rho) \bm I\right) \bm H,
\end{equation}
where the coordinate system has been rotated such that $\bm \Sigma_Z =
\bm I$. The MAF coefficients will change but the resulting MAF factors
will not under this rotation, as shown in Lemma \ref{lem:transform}.

 Now, let $\bar{\bm u}_i$ be the mean of the
normalized version of the vectors in \eq{eq:35}, $\lambda_i$ be the
corresponding eigenvectors. By
letting $\Lambda$ be the diagonal matrix with
$\frac{(1 - \rho)}{\sqrt{\lambda_i}}$ along the diagonal and
$c_i = \frac{\rho p\bar{\bm u}_i}{\sqrt{\lambda_i}}$, we can
recast this in a more familiar form,
\begin{equation}
\label{simpleForm}
  (\tilde{\Sigma}_{\Delta})_{ij} = \bm \Lambda + \bm c \bm c'.
\end{equation}

A closer look at this matrix will reveal that $c_i = 0, \, \forall
i>2$. Furthermore, $\bm \Lambda_{ii} = 1, \, \forall i>2$. This means that we can
decompose the matrix as follows,
\begin{equation}
  \label{eq:matrix}
  \begin{bmatrix}
    \bm A & \bm 0_{2\times (p-2)} \\
    \bm  0_{(p-2) \times 2} & \bm 1_{(p-2)\times (p-2)}.
  \end{bmatrix}
\end{equation}

And because $\bm A$ is symmetric and $2\times 2$, its eigenvalues/eigenvectors can be
found in closed form. The remaining eigenvectors can be made the
standard basis vectors $\bm e_i = (0_{1},..., 0_{i-1}, 1, 0_{i+1}, ...,
0_{p}), \, \forall i>2$. In particular, by solving
\begin{equation}
  \label{eq:eigen2by2}
  \bm A \begin{bmatrix} x \\ y \end{bmatrix} = \begin{bmatrix}
    \lambda^{-1}_1\left( 1-\rho + \rho p^2\bar{\bm u}^2_1\right) &
    \frac{\rho p^2}{\sqrt{\lambda_1\lambda_2}}\bar{\bm u}_1\bar{\bm u}_2 \\
    \frac{\rho p^2}{\sqrt{\lambda_1\lambda_2}}\bar{\bm u}_1\bar{\bm u}_2
    & \lambda^{-1}_2\left( 1-\rho + \rho p^2\bar{\bm u}^2_2\right) \end{bmatrix}\begin{bmatrix} x \\ y \end{bmatrix} = \begin{bmatrix} a & b \\ b & d \end{bmatrix}\begin{bmatrix} x \\ y \end{bmatrix} =
  \mu \begin{bmatrix} x \\ y \end{bmatrix},
\end{equation}
we find that the eigenvalues/eigenvectors are
\begin{equation}
  \label{eq:2by2eigen}
  \begin{bmatrix} x\\ y\end{bmatrix} = \frac{1}{\sqrt{b^2 + (\mu -
    d)^2}} \begin{bmatrix} \mu - d \\ b \end{bmatrix},  \quad
   \mu_1, \mu_2 = \frac{a+d \pm \sqrt{(a-d)^2 + 4b^2}}{2},
\end{equation}
where we are interested in the smallest eigenvalue, i.e. where we
subtract the term involving the discriminant.

Now, let $\tilde{\bm w}$ be the full vector in $\R^p$ with zeroes everywhere except
in the first two entries which take the values $x$ and $y$. To get the
values of each coefficient in the basis of the original time series, we do an inverse transformation,
\begin{equation}
  \label{eq:inverseBasisAppendix}
  \bm w = (\bm H')^{-1} \tilde{\bm w} = \bm H \tilde{\bm w} = \begin{bmatrix}  \bm h_1 & \bm h_2 &
    \ldots &
    \bm h_p \end{bmatrix} \begin{bmatrix} x \\ y \\ 0 \\ \vdots \\ 0 \end{bmatrix}= x \bm h_1 + y \bm h_2.
\end{equation}

Note that this expression is a linear combination of PC1 and
PC2. Furthermore, the largest eigenvalue of the two in
\eq{eq:2by2eigen} is equal to 1, just like the other degenerate
eigenvalues. This leaves only one non-degenerate eigenvalue, which
can be interpreted as there being only one signal present in different
strenths.

The result in \eq{eq:inverseBasisAppendix} can be used to obtain MAF1 in a more general setting,
where the noise is of unequal variance,
\begin{equation}
  \label{eq:41}
  \bm z_t = \bm b f_t + \bm \varepsilon_t, \: \: \text{with}
  \: \: (\varepsilon_t)_i \sim (0, \sigma_i^2) \quad \forall t,
\end{equation}
 by taking advantage of the
fact that MAF is preserved under linear transformations.
We
can write the modified covariance matrix as
\begin{equation}
  \label{eq:222}
\bm A' \bm\Sigma_Z\bm A = \bm A' \left( \tilde {\bm b} \tilde{\bm b}'
  + \rho \bm 1_p\bm 1_p' + (1-\rho)\bm I
\right) \bm A ,
\end{equation}
where $\tilde{\bm b}= \bm b/{\bm \sigma}$ and $\bm \sigma^2$ is the
vector is noise variances. Similarly for the lagged covariance matrix.

It then follows that the eigenvalue equation to be solved is
\begin{equation}
  \label{eq:40}
  \bm \Sigma_Z^{-1/2} \bm \Sigma_\Delta \bm \Sigma_Z^{-1/2} \tilde{\bm w} = \lambda \tilde{\bm w},
\end{equation}
where $\tilde{\bm w} = \bm A \bm w $. But $\tilde{\bm w}$ is already given in
\eqref{eq:inverseBasisAppendix}, and thus $\bm w = \bm A^{-1} \tilde{\bm w}$, which are the
MAF1 coefficients in the original coordinate system.

This this general case with unequal noise variance, the MAF will not be a linear
combination of PC1 and PC2. The reason for this is that the
vectors in the nullspace of the new covariance matrix' rank-2
update will not in general be an eigenvector of the full covariance
matrix, with the unequal variance terms in the diagonal. This means
that the eigenvalue problem cannot be rewritten in a form similar to
the one given in \eqref{eq:matrix}. In fact, the PCA eigenvectors do
not even exist in closed form, but must be obtained by solving a determinant
equation for the eigenvalues. This problem is explored in
\citet{Arbenz1988a, Bunch1978} and the references therein.


\section{Proofs}
\label{app:proofs}

\begin{proof}[Proof of Lemma \ref{lem:expLikeli}]
Under $H_0$ and the Gaussian assumption, the log-likelihood
\begin{equation}
  \label{eq:13a}
  \ln(L_0) = -\frac{1}{2}\ln(|\bm \Sigma_\varepsilon |) - \frac{1}{2}\sum_{t=1}^n\bm z(t)'\bm \Sigma_\varepsilon^{-1}\bm z(t) - \frac{p}{2}\ln(2\pi),
\end{equation}
where we assume that the mean of $\bm z_t$ is zero without loss of generality.

Now, the likelihood under the alternative hypothesis
\begin{equation}
  \label{eqa:14}
  \ln(L_A) = -\frac{1}{2}\ln(|\bm \Sigma_\varepsilon |) - \frac{1}{2}\sum_{t=1}^n(\bm z(t) - \bm
  b f(t))'\bm \Sigma_\varepsilon^{-1}(\bm z(t) - \bm
  b f(t)) - \frac{p}{2}\ln(2\pi).
\end{equation}

The likelihood ratio is
\begin{align}
  \label{likelihoodRatio}
\ln(L_A) - \ln(L_0) = &-\frac{1}{2}\sum_{t=1}^nf^2(t)\bm b \bm \Sigma_\varepsilon^{-1}
\bm b + \sum_{t=1}^nf(t)\bm z(t)'\bm \Sigma_\varepsilon^{-1}\bm b
\nonumber \\
 = &- \frac{1}{2}\bm b' \Sigma_\varepsilon^{-1} \bm b + \sum_{t=1}^nf(t)\bm z(t)'\bm \Sigma_\varepsilon^{-1}\bm b
\end{align}
where the last equality holds because of the unit
squared sum of $f(t)$. Now, if we substitute
$\bm z(t)$ for the alternative model and taking expectations under
either model, we get
\begin{align}
  \label{likelihoodRatio2}
E[\ln(L_A) - \ln(L_0)] = & -\frac{1}{2}\bm b' \Sigma_\varepsilon^{-1}
\bm b + E\left[\sum_{t=1}^nf(t)(\bm b f(t) - \bm\varepsilon(t))'\bm
  \Sigma_\varepsilon^{-1}\bm b \right]
\nonumber \\
 = &  \frac{1}{2}\bm b' \Sigma_\varepsilon^{-1}
\bm b,
\end{align}
where the term involving $\bm\varepsilon$ is zero after taking the
expectation.

Now, we already proved that MAF1 maximizes the model signal-to-noise
ratio,
\begin{equation}
  \label{eq:18}
  SNR = \frac{(\bm w' \bm b)^2 }{\bm w' \bm \Sigma_\varepsilon \bm w}.
\end{equation}

Making the change of coordinates $\bm u = \Sigma_\varepsilon^{1/2} \bm w$, using the
familiar spectral decomposition for the square root, gives the SNR
representation
\begin{equation}
  \label{eq:23}
  \frac{(\bm u'\bm \Sigma_\varepsilon^{-1/2} \bm b)^2 }{\bm u'  \bm u}.
\end{equation}

Using Cauchy-Schwartz theorem, the normalized vector, $\bm u$, which maximizes SNR
is parallel to $\bm \Sigma_\varepsilon^{-1/2} \bm b$. And so in the original
coordinate system,
\begin{equation}
  \label{eq:24}
  \bm w_{MAF} = \bm \Sigma_\varepsilon^{-1}\bm b
\end{equation}.

Substituting the expression for these MAF1 coefficients in our
definition for SNR gives
\begin{align}
  SNR_{optimal} = &\frac{(\bm b' \bm \Sigma_\varepsilon^{-1} \bm b)^2}{\bm b'
    \bm\Sigma_\varepsilon^{-1}\bm\Sigma_\varepsilon\bm\Sigma_\varepsilon^{-1} \bm b} \nonumber \\
= & \bm b' \bm \Sigma_\varepsilon^{-1} \bm b,
\end{align}
which is proportional to the expected likelihood ratio test statistic.
\end{proof}

\begin{prop}
\label{propMAF}
The MAF-transformation matrix, $\bm W_{MAF} = [\bm w_1, \bm w_2, ...,
\bm w_p]$ , contains the eigenvectors
of $\bm S^{-1/2} \bm S_{\Delta} \bm S^{-1/2}$.
\end{prop}
\begin{proof}[Proof of Proposition \ref{propMAF}]
By definition, the MAF1 vector of coefficients, $\bm w_1$, minimizes
\begin{equation}
  \label{eq:optCoeffs}
  \frac{\bm w_1' \bm S_{\Delta} \bm w_1}{\bm w_1' \bm S \bm w_1}.
\end{equation}

By letting $\bm S^{1/2} \bm w_1 = \bm u_1$ we get
\begin{equation}
  \label{eq:15b}
  \frac{\bm u_1' \bm S^{-1/2} \bm S_{\Delta} \bm S^{-1/2}\bm u_1}{\bm u_1'
    \bm u_1},
\end{equation}
which is equivalent to minimizing
\begin{align}
 & \bm u_1' \bm S^{-1/2} \bm S_{\Delta} \bm S^{-1/2}\bm u_1
 \nonumber \\
 \text{subject to} \quad
 & \bm u_1'\bm u_1 = 1.
\end{align}

Following \citet{muir}, the minimizing vector is the eigenvector with
the lowest eigenvalue. Furthermore, the eigenvector with the $k^{th}$
smallest eigenvalue minimizes
\begin{align}
 & \bm u_k' \bm S^{-1/2} \bm S_{\Delta} \bm S^{-1/2}\bm u_k
 \nonumber \\
 \text{subject to} \quad
 & \bm u_k'\bm u_k = 1 \nonumber\\
 \text{and} \quad & \bm u_k'\bm u_i = 0 \quad \forall i<k.
\end{align}

$\bm u_k$ corresponds to MAF$k$, after MAF$i, \: \forall i < k$, has been
projected out of the data. The linear transformation $\bm w_k = \bm
S^{-1/2} \bm u_k$ gives the eigenvectors in the original coordinate system.
\end{proof}

\newpage
\newtheorem*{nonumberThm}{Theorem \ref{theorem1a} (Simplified)}
\begin{nonumberThm}
\label{theorem1a_simplified}
Consider a set of time series $Z_n(t) \in \R^p$, such that
\begin{align}
  \label{eq:21mod2}
  \bm Z_n(t) =& f_n(t) \bm b + \bm \varepsilon_n(t) \quad t= 1, \dots, n
  \nonumber \\
  \bm \Delta \bm Z_n(t) =& \bm Z_n(t) - \bm Z_n(t+1) =  \Delta f_n(t) \bm b + \Delta \bm \varepsilon_n(t)
\end{align}
with $f_n(t) \in \R \:\: \forall t = 1,2,...,n , \bm b, \bm \varepsilon_n(t)  \in \R^p\:\: \forall t = 1,2,...,n $, $\Delta
\bm \varepsilon_n = \bm \varepsilon_n(t) - \bm \varepsilon_n(t+1)$,
and $\Delta f_n(t) = f_n(t) - f_n(t+1)$. Residual time series $\bm\varepsilon_n$ is a
weakly stationary $p$-variate time series and the associated autocovariance is
absolutely summable.
 Then,
 \begin{align}
 \bm S_{n} \overset{p}{\to}&  \bm\Sigma \text{  as }n \to \infty, \nonumber\\
\bm S^{-1/2} \bm S_\Delta \bm S^{-1/2} \overset{p}{\to}& \bm \Sigma^{-1/2}  \bm \Sigma_\Delta \bm \Sigma^{-1/2} \text{  as }n \to \infty \nonumber\\
\end{align}
\end{nonumberThm}

\begin{proof}[Proof of Theorem \ref{theorem1a}]
Three parts make up this proof: 1. Stationarity of differenced
time series, 2. Convergence in probability of $\bm S_n$ and $ \bm S_{\Delta n}$, the
sample cross-correlation and lagged cross-correlation to their model counterparts.
3. Consistency of MAF and PCA coefficients to their model counterparts.

We begin by recognizing that since $\bm\varepsilon_n$ is a
weakly stationary $p$-variate time series, we have
\begin{align}
  E[\bm \varepsilon_n (t)] = &\bm 0 \nonumber \\
  Cov[\bm \varepsilon_n(t)] = &\bm\Sigma_\varepsilon \nonumber \\
  Cov[\Delta \bm \varepsilon_n(t)] = &\bm\Sigma_{\Delta\varepsilon},
\end{align}
 with the assumption of lagged summability,
\begin{equation}
  \label{eq:20}
  \sum_{\tau = 0}^{\infty} |\bm \gamma_{\tau, i} |< \infty.
\end{equation}
where the lagged autocovariance for noise component $i$ is
 \begin{equation}
   \label{eq:19}
   \bm \gamma_{\tau, i} = Cov[\varepsilon_{n,i}(t), \varepsilon_{n, i}(t+\tau)].
 \end{equation}

Furthermore, let
\begin{align}
  \label{eq:Sn}
 \bm S_n =& \frac{1}{n}\sum_{t=1}^n[\bm Z_n(t) - \overline{\bm Z_n(t)}]
  [\bm Z_n(t) - \overline{\bm Z_n(t)}]' \nonumber \\
  \bm S_{\Delta n} =& \frac{1}{n}\sum_{t=1}^n\bm [\Delta \bm Z_n(t) -
  \overline{\Delta \bm Z_n(t)}] [\Delta \bm Z_n(t) - \overline{\Delta \bm Z_n(t)}]',
\end{align}
where $\overline{\bm Z_n(t)} = \sum_{i=1}^n \bm Z_n(t)$, similarly for
$\Delta \bm Z_n(t)$.
\newline\newline
\textbf{Part I: Stationarity of differenced time series:}\newline
We now show that
$\Delta
\bm\varepsilon_n(t) - \overline{\Delta \bm\varepsilon_n}$ is a zero
mean weakly stationary time series, using the weak stationarity of
$\bm \varepsilon_n(t)$. Let
\begin{equation}
  \label{eq:21}
  Cov[\bm\varepsilon(t), \bm\varepsilon(t+h)] = \bm\gamma(h),
\end{equation}
which is by definition not a function of $t$. Then
\begin{align}
  \label{eq:22}
Cov[\Delta\bm\varepsilon(t), \Delta\bm\varepsilon(t+h)] = &
Cov[\bm\varepsilon(t), \bm\varepsilon(t+h)] - Cov[\bm\varepsilon(t+1),
\bm\varepsilon(t+h)] + \nonumber \\ & \quad Cov[\bm\varepsilon(t+1), \bm\varepsilon(t+h+1)]
- Cov[\bm\varepsilon(t), \bm\varepsilon(t+h+1)] \nonumber \\
= & 2 \bm\gamma(h) - \bm\gamma(h+1) - \bm\gamma(h-1) \quad \text{for }
h > 0,
\end{align}
is not a function of $t$ and is thus also weakly stationary. It is
trivial to show that the differenced time series has zero mean.
\newline\newline
\textbf{Part II: Consistency of $\bm S_n$ and $\bm S_{\Delta n}$ }\newline
Substituting \eq{eq:21mod} in \eq{eq:8}, we get
\begin{align}
\bm S_n =&
  \frac{1}{n} \sum_{t=1}^n [(f_n(t) - \overline{f_n})\bm b +
  \bm\varepsilon_n(t) - \overline{\bm\varepsilon_n}][(f_n(t) - \overline{f_n})\bm b +
  \bm\varepsilon_n(t) - \overline{\bm\varepsilon_n}]' \nonumber \\
=&  \frac{1}{n} \sum_{t=1}^n [\bm b \bm b' f_n^2(t) + 2 f_n(t) \bm b
(\bm\varepsilon_n(t) - \overline{\bm\varepsilon_n})' +
(\bm\varepsilon_n(t) -
\overline{\bm\varepsilon_n})(\bm\varepsilon_n(t) - \overline{\bm\varepsilon_n})'].
\end{align}

As $n \to \infty$, The first term equals $\bm b \bm b'$ by definition, the second term
goes to zero in probability because the $p$-vector of the
time-averaged residuals goes to the zero vector in probability, and the third term goes to
$\bm \Sigma_\varepsilon$ because $\bm\varepsilon_n(t)$ is a weakly stationary
time series with zero mean \citep{Doob1953, Durrett2010}.

Thus,
\begin{equation}
  \label{eq:3}
  \bm S_n =\frac{1}{n} \sum_{t=1}^n [\bm Z_n(t) - \overline{\bm Z_n(t)}]
  [\bm Z_n(t)  - \overline{\bm Z_n(t)}]' \overset{p}{\to} \bm b \bm
  b' + \bm \Sigma_\varepsilon = E[\bm S_n]
\end{equation}

Now, consider
\begin{align}
  \label{eq:Sdn}
   \bm S_{\Delta n} =&  \frac{1}{n} \sum_t [\bm b \bm b' (\Delta f_n -
   \overline{\Delta f_n})^2(t) + 2 (\Delta f_n(t) - \overline{\Delta f_n}) \bm b
( \Delta \bm\varepsilon_n(t) - \overline{\Delta \bm\varepsilon_n})'(t) +
(\Delta \bm\varepsilon_n(t) - \overline{\Delta
  \bm\varepsilon_n})(\Delta \bm\varepsilon_n(t) - \overline{\Delta \bm\varepsilon_n})'].
\end{align}
Applying the same arguments and using the
weak stationarity of the differenced time series, we get
\begin{equation}
  \label{eq:3p}
  \bm S_{\Delta n} = \frac{1}{n} \sum_t [\Delta\bm Z_n(t) -
  \overline{\Delta\bm Z_n(t)}] [\Delta\bm Z_n(t) - \overline{\Delta\bm Z_n(t)}]' \overset{p}{\to} a\bm b \bm
  b' + \bm \Sigma_{\Delta\varepsilon} = E[\bm S_{\Delta n}].
\end{equation}

To wit,
\begin{align}
\bm S_n \overset{p}{\to}& E[\bm S_n] = \bm\Sigma = \bm b \bm b' + \bm \Sigma_\epsilon \text{  as }n \to \infty\nonumber\\
  \bm S_{\Delta n} \overset{p}{\to}& E[\bm S_{\Delta n}] = \bm\Sigma_\Delta = a\bm b \bm b' +
  \bm \Sigma_{\Delta\epsilon} \text{  as }n \to \infty,
\end{align}
\newline\newline
\textbf{Part III: Consistency of MAF and PCA coefficients}\newline
PCA and MAF coefficients are the eigenvectors of $\bm S$ and $\bm S^{-1/2} \bm S_\Delta \bm S^{-1/2}$ respectively, using a spectral
decomposition $\bm S^{-1/2} = \bm H \bm L^{-1/2} \bm H'$. The continuous
mapping theorem ensures that consistent estimates of the covariance
and lagged covariance matrices implies consistent estimates of MAF and
PCA coefficients, i.e. the coefficients will also converge to their
model values, the eigenvectors of the model covariance matrix, $\bm \Sigma$ for PCA
and $\bm \Sigma^{-1/2}  \Sigma_\Delta \bm \Sigma^{-1/2}$, with $\bm
\Sigma^{-1/2} = \bm\Gamma \bm D^{-1/2} \bm\Gamma'$.
\end{proof}

\begin{proof}[Proof of Theorem \ref{multipleSignalsTheorem}]

First we establish that MAF coefficient vectors $1$ to $q$ are linear combinations of the signal strength vectors $\bm b_i$. Then we show that CCA coefficients also has this property. Thus, these two methods have coefficients that span the same $q$-subspace in $\R^{p}$.

\textbf{Part I: MAF coefficient vectors}\newline
The definition of the first $q$ MAF factors solve the following sequence of problems,
\begin{align}
\label{conevxEqNoniid}
  \text{maximize} &
       \quad\rho(\bm a_i) = \frac{\bm a_i'\bm \Sigma_{\delta Z} \bm a_i}{\bm a_i'\bm \Sigma_Z \bm a_i} = k_\varepsilon \frac{\bm B \bm \Lambda \bm B' + \bm I}{\bm B  \bm B' + \bm I}\nonumber\\
  \text{subject to} & \quad  \bm a_i'\bm a_i = 1 \nonumber\\
   & \quad \bm a_i'\bm a_j = 0, \quad \forall j < i \leq q,
\end{align}
where we have changed coordinate system such that $\bm \Sigma_\varepsilon = \bm I$ and where $\Lambda$ is the matrix with entries $(k_i)/k_\varepsilon = \lambda_i$ for $i = 1, \dots, q$ in the diagonal and zero in the off-diagonals. For now assume that the columns of $\bm B$ are linearly independent.

We want to show that the first $q$ maximizing vectors are in the range of $\bm B$.

First, from the spectral theorem
\begin{equation}
  \bm B \bm B' = \bm U \bm D \bm U',
\end{equation}
where $U \in \R^{p \times p}$ whose last $p-q$ columns are perpendicular to $\bm b_i$ $\forall i = 1, \dots, q$.

So,
\begin{equation}
  \bm B \bm B' + \bm I = \bm U (\bm D + \bm I) \bm U',
\end{equation}
so the eigenvectors are preserved by adding the identity matrix.

Second, note that the first $q$ columns of $\bm U$, $\bm u_i \in R(\bm B)$ since we can write
\begin{equation}
  \bm B [\bm B' \bm U \bm D^{-1}] = \bm U.
\end{equation}

Third, note that for any $\bm x \in R(\bm B)$
\begin{equation}
\label{starEq}
  0 \leq \sum_{i=1}^q (\bm b_i'\bm x)^2 = \bm x \bm B \bm B' \bm x \leq \sum_{i=1}^q \lambda_i (\bm b_i'\bm x)^2 = \bm x \bm B \bm \Lambda \bm B' \bm x,
\end{equation}
and thus
\begin{equation}
\label{ineqA}
  \bm x' (\bm B \bm B' + \bm I) \bm x \leq \bm x' (\bm B \bm \Lambda \bm B' + \bm I) \bm x,
\end{equation}
with strict inequality iff $\bm x \in R(\bm B)$.

Now,
\begin{equation}
  \bm B \bm \Lambda \bm B' + \bm I = \bm V (\bm E + \bm I) \bm V',
\end{equation}
and since \eq{ineqA} holds for any vector in the range of $\bm B$, any eigenvector in $\bm U$ or $\bm V$ with eigenvalue greater than 1 must be in the range of $\bm B$.
So,
\begin{equation}
  \rho(\bm x) = k_\varepsilon \frac{\bm x \bm V (\bm E + \bm I) \bm V' \bm x}{\bm x \bm U (\bm D + \bm I) \bm U' \bm x} \geq k_\varepsilon.
\end{equation}

We can thus find $q$ linearly independent vectors in $R(\bm B)$, call them $\bm x_B$ such that $\rho(\bm x_B) > k_\varepsilon$. Any vectors in the null space of $B$ will have $\rho(\bm x_B) = k_\varepsilon$, and so these will appear after the set of vectors in $R(\bm B)$. If the rank of $\bm B \bm B' < q$ the proof follows the same arguments with a lower dimension substituted for $q$. In the original co-ordinate system where $\bm \Sigma_\varepsilon \neq \bm I$, the vectors $\bm a_i$ will be in the range of $\bm \Sigma_\varepsilon^{-1} \bm B$ as seen by a change of coordinate transform $\bm a = \tilde{\bm a} \bm \Sigma_\varepsilon^{-1/2}$.

 \textbf{Part II: CCA coefficient vectors}\newline

Now if we consider Canonical Correlation Analysis (CCA), we look for
the linear combination of the columns of $\bm Z$ which maximizes a linear
combination of the signals $F = (f_1(t), f_2(t), ..., f_k(t))$. Without loss of generality,
\begin{equation}
\label{eq:FI}
\Sigma_F = \bm I,
\end{equation}
and thus
\begin{equation}
\label{eq:ZF}
\Sigma_{ZF} = \bm B
\end{equation}

By definition the canonical variables of Z have linear weights given by the eigenvectors of
\begin{align}
 \Sigma_Z^{-1}\Sigma_{ZF}\Sigma_{F}^{-1}\Sigma_{FZ} \nonumber\\
 \end{align}
which, using \eq{eq:FI} and \eq{eq:ZF} and $\Sigma_F = \bm I$ become
\begin{align}
  \label{eq:6a}
  (\bm B \bm B' + \Sigma_\varepsilon)^{-1} \bm B \bm B'\nonumber\\
 \end{align}
which is equivalent to the sequence of problems
\begin{align}
\label{conevxEqNoniidCCA}
  \text{maximize} &
       \quad\phi(\bm a_i) = \frac{1}{1+\frac{1}{\sigma(\bm a_i)}} \nonumber\\
  \text{subject to} & \quad  \bm a_i'\bm a_i = 1 \nonumber\\
   & \quad \bm a_i'\bm a_j = 0, \quad \forall j < i < q.
\end{align}

But this is equivalent to

\begin{align}
\label{conevxEqNoniidCCA2}
  \text{maximize} &
       \quad \sigma(\bm a_i) = \frac{\bm a_i'\bm B \bm B\bm a_i}{\bm a_i'\bm \Sigma_\varepsilon \bm a_i}\nonumber\\
  \text{subject to} & \quad  \bm a_i'\bm a_i = 1 \nonumber\\
   & \quad \bm a_i'\bm a_j = 0, \quad \forall j < i < q.
\end{align}

Now, making the change of variable $\tilde{\bm a_i} = \bm \Sigma_\varepsilon^{1/2}a_i$
$\tilde{\bm B} = \bm\Sigma_\varepsilon^{-1/2}\bm B$ our CCA problem reduces to finding the eigenvalues of $\tilde{\bm B}\tilde{\bm B}'$, sorted in the diagonal matrix $\tilde{\bm D} = diag(\tilde{d}_1, \tilde{d}_2, \dots, \tilde{d}_q)$. By the spectral theorem
\begin{align}
\tilde{\bm B}\tilde{\bm B}' = & \tilde{\bm A}\tilde{\bm D}\tilde{\bm A} \nonumber \\
\bm A = &\bm\Sigma_\varepsilon^{-1} \bm B \bm B' \bm A \tilde{\bm D}^{-1} \nonumber\\
\bm A = & \bm\Sigma_\varepsilon^{-1} \bm B \bm U,
\end{align}
where $\bm U = \bm B' \bm A \tilde{\bm D}^{-1}$. Now we see that the eigenvectors are linear combinations of $\bm\Sigma_\varepsilon^{-1} \bm b_i, \quad \forall i \leq q$.
\end{proof}

\begin{proof}[Proof of Property \ref{pr:nonoise}]
A linear combination, $\bm a'\bm Z(t)$, of the $p$-vector time series $\bm Z(t)$ can be expressed as $\bm a'\bm Z(t) = \bm b'\bm f(t)$, where $\bm b'=\bm a' \bm B$, $\bm Z(t) = \bm B\bm f(t)$, $f(t)$ is the unknown underlying $q$-vector factor time series, and A is the unknown $p\times q$ loading matrix.  Let $r(\bm a)$ denote the unit-lag autocorrelation of the scalar time series $\bm a'\bm Z(t)$.  Then the orthogonality of the factors $\bm f(t)$ yields $r(\bm a) = \bm b'diag(\bm k)\bm b / \bm b'\bm b$ where $r$ is the $q$-vector of the underlying factor autocorrelations in decreasing order. Since $r(\bm a)$ is therefore a convex combination of the factor autocorrelations, $R(\bm a)$ cannot be greater than the largest of the underlying factor autocorrelations, $r_1$.  Therefore MAF-1, which maximizes $r(\bm a)$, yields factor loadings $b_i=0$ for $i>1$, and a MAF-1 time series $\bm a'\bm Z(t)$ is proportional to the underlying factor-1 time series $f_1(t)$.  For $q<p$, the MAF-1 optimizing coefficient $\bm a'$ for the linear combination of $\bm a'\bm Z(t)$ is not unique because the $p\times p$ cross-covariance of the vector time series $\bm Z(t)$ has rank $q$. But every linear combination that maximizes autocorrelation will yield a time series that is proportional to the underlying factor time series $f_1(t)$.  Similarly, successive orthogonal MAF time series, up to MAF-$q$, will evaluate to the corresponding successive underlying factor time series $f_2(t), ..., f_q(t)$.
\end{proof}


\end{document}